\documentclass[12pt,oneside]{amsart}

\usepackage{amsthm,amsmath,amssymb,amsfonts,amscd}
\usepackage{amsaddr}

\newcommand{\cS}{{\mathcal{S}}}

\newcommand{\mnf}{\mathbb{M}}

\newcommand{\fap}{{\textbf{f}}}
\newcommand{\gap}{{\textbf{g}}}

\newcommand{\wor}[1]{\fik^{#1}\setminus\{0\}}

\newcommand{\rs}[1]{\mathcal{R}_{#1}}
\newcommand{\gq}[3]{\mathcal{G}(#1,#2,#3)}

\newcommand{\hg}[1]{\mathbb{H}_{#1}}

\newcommand{\ud}{u_l,\dots}
\newcommand{\ua}{{\boldsymbol\alpha}}
\newcommand{\ub}{{\boldsymbol\beta}}
\newcommand{\tua}{{\tilde{\boldsymbol\alpha}}}
\newcommand{\tub}{{\tilde{\boldsymbol\beta}}}

\newcommand{\ta}{{a}}
\newcommand{\tb}{{b}}
\newcommand{\fe}{F}

\DeclareMathOperator{\codim}{codim}

\newcommand{\prs}{\mathbb{CP}}

\newcommand{\iel}{\mathbf{1}}

\newcommand{\fv}{v}

\newcommand{\cp}{c}
\newcommand{\og}{m}
\newcommand{\mf}{f}
\newcommand{\la}{\lambda}

\newcommand{\ff}{\mathbf{F}}
\newcommand{\tff}{\tilde{\mathbf{F}}}

\newcommand{\gl}{\mathrm{Mat}}
\newcommand{\GL}{\mathrm{GL}}
\newcommand{\gla}{\mathbf{G}}
\newcommand{\fik}{\mathbb{C}}

\newcommand{\sm}{d}
\newcommand{\pmm}{x}

\newcommand{\lb}{\label}
\newcommand{\ga}{\gamma}
\newcommand{\vf}{\varphi}
\newcommand{\zp}{\mathbb{Z}_{\ge 0}}
\newcommand{\zsp}{\mathbb{Z}_{>0}}
\newcommand{\zz}{\mathbb{Z}}

\newcommand{\pd}{\partial}
\newcommand{\lt}{\mathcal{U}}
\newcommand{\al}{\alpha}
\newcommand{\be}{\beta}
\newcommand{\diu}{N}
\newcommand{\ti}{{\tilde\imath}}
\newcommand{\tj}{{\tilde\jmath}}
\newcommand{\tq}{{\tilde q}}

\newcommand{\tk}{{\tilde k}}
\newcommand{\hi}{{\hat\imath}}
\newcommand{\hk}{{\hat k}}
\newcommand{\ci}{{\check\imath}}
\newcommand{\ck}{{\check k}}
\newcommand{\da}{{\textbf{d}}}

\newtheorem{theorem}{Theorem}
\newtheorem{lemma}{Lemma}

\theoremstyle{definition}
\newtheorem{definition}{Definition}
\newtheorem{example}{Example}
\newtheorem{remark}{Remark}

\begin{document}

\title[Miura-type transformations and Lie group actions]{Miura-type transformations 
for lattice equations and Lie group actions associated with Darboux-Lax representations}
\date{}

\author{George Berkeley\qquad\qquad Sergei Igonin} 
\address{School of Mathematics, University of Leeds, LS2 9JT Leeds, UK}

\author{}
\email{george231086@yahoo.co.uk, s-igonin@yandex.ru}



\begin{abstract}
Miura-type transformations (MTs) are an essential tool in the theory of 
integrable nonlinear partial differential and difference equations.
We present a geometric method to construct MTs 
for differential-difference (lattice) equations 
from Darboux-Lax representations (DLRs) of such equations.

The method is applicable to parameter-dependent DLRs satisfying certain conditions. 
We construct MTs and modified lattice equations 
from invariants of some Lie group actions on manifolds associated with such DLRs.

Using this construction, from a given suitable DLR one can obtain many MTs of different orders. 
The main idea behind this method is closely related to the results of 
Drinfeld and Sokolov on MTs for the partial differential KdV equation.

Considered examples include the Volterra,
Narita-Itoh-Bogoyavlensky, Toda, and Adler-Postnikov lattices.
Some of the constructed MTs and modified lattice equations seem to be new. 

\end{abstract}

\maketitle

\section{Introduction}

\subsection{An overview of the results}

It is well known that Miura-type transformations (MTs) 
play an essential role in the theory of 
integrable nonlinear partial differential and difference equations.
(For partial differential equations, MTs are also called differential substitutions.)

In particular, when one tries to classify a certain 
class of integrable equations, one often finds a few basic equations 
such that all other equations from the considered class can be obtained from 
the basic ones by means of MTs (see, e.g.,~\cite{mss91,yam2006,gram11,meshk2008} 
and references therein).
Applications of MTs to construction of conservation laws~\cite{miura-cl,suris03} and 
auto-B\"acklund transformations are also well known.

Therefore, it is highly desirable to develop systematic methods for constructing MTs.
In this paper, we focus on MTs of differential-difference (lattice) equations.

Let $\ua$, $\ub$ be integers such that $\ua\le\ub$.
We study differential-difference equations of the form
\begin{equation}
\lb{isdde}
u_t=\ff(u_\ua,u_{\ua+1},\dots,u_\ub),
\end{equation}
where $u=u(n,t)$ is a vector-function of
an integer variable~$n$ and a real or complex variable~$t$.
We use the standard notation $u_t=\pd_t(u)$ and $u_l=u(n+l,t)$ for $l\in\zz$.
In particular, $u_0=u$.

Equation~\eqref{isdde} must be valid for all $n\in\zz$, 
so~\eqref{isdde} encodes an infinite sequence of differential equations
$$
\pd_t\big(u(n,t)\big)=\ff\big(u(n+\ua,t),u(n+\ua+1,t),\dots,u(n+\ub,t)\big),\qquad\quad n\in\zz.
$$

Consider another differential-difference equation of similar type
\begin{equation}
\lb{ivdde}
v_t=\tff(v_\tua,v_{\tua+1},\dots,v_\tub)
\end{equation}
for a vector-function $v=v(n,t)$ and some integers $\tua\le\tub$.

A \emph{Miura-type transformation} (MT) from equation~\eqref{ivdde} to equation~\eqref{isdde} 
is determined by an expression of the form
\begin{equation}
\lb{iuvf}
u=\mf(v_p,v_{p+1},\dots,v_r),\qquad\quad p,r\in\zz,\qquad p\le r,
\end{equation}
such that 
if $v$ satisfies~\eqref{ivdde} then $u$ given by~\eqref{iuvf} satisfies~\eqref{isdde}.
A more precise definition of MTs is presented in Section~\ref{ddelr}.

Equation~\eqref{ivdde} is called the \emph{modified equation} 
corresponding to the described MT.

Let $p,r\in\zz$, $p\le r$, be such that the function $\mf$ in~\eqref{iuvf} 
may depend only on $v_p,v_{p+1},\dots,v_r$ and depends nontrivially on $v_p$, $v_r$.
Then the number $r-p$ is called the \emph{order} of the MT~\eqref{iuvf}.

\begin{example}
Let $u$ and $v$ be scalar functions.
It is known that the formula $u=v v_1$ determines an MT 
from the modified Volterra equation $v_t=v^2 (v_1-v_{-1})$
to the Volterra equation $u_t=u(u_1-u_{-1})$.

According to our notation, $v=v_0$, so $v v_1=v_0 v_1$. 
Therefore, the MT $u=v v_1$ is of order~$1$.
\end{example}

We present a method to construct equations~\eqref{ivdde} and MTs~\eqref{iuvf} of different orders 
for a given equation~\eqref{isdde} possessing a matrix Darboux-Lax representation (DLR).
Note that the order of the obtained MTs may be higher   
than the size of the matrices in the DLR. 
(The definition of DLRs is given in Remark~\ref{idlr} below.)

The method uses invariants of some Lie group actions on manifolds associated with DLRs and 
is applicable to parameter-dependent DLRs satisfying certain conditions, 
see Section~\ref{secmi} and Section~\ref{secgt} for more details.
The main idea behind the method is closely related to the results of 
Drinfeld, Sokolov~\cite{drin-sok85} and Igonin~\cite{igon05}
on MTs for evolution partial differential equations (PDEs).

The papers~\cite{drin-sok85,igon05} construct MTs for evolution PDEs 
from zero-curvature representations of evolution PDEs.
We construct MTs for differential-difference equations from 
Darboux-Lax representations of differential-difference equations.
Since the structure of Darboux-Lax representations 
of differential-difference equations is quite different from the structure of zero-curvature representations of PDEs, 
translation of ideas of~\cite{drin-sok85,igon05} to the case of differential-difference equations is a challenging problem.

Note that the papers~\cite{drin-sok85,igon05} study properly only the case 
of scalar equations.
More precisely, the paper~\cite{drin-sok85} presents (without proof) 
a construction of MTs for the (scalar) partial differential KdV equation from the standard 
zero-curvature representation of KdV.
A similar (but more general) construction of MTs for scalar $(1+1)$-dimensional evolution PDEs is given in~\cite{igon05}.
The multicomponent case is not discussed at all in~\cite{drin-sok85}.
In addition to scalar evolution PDEs, the paper \cite{igon05} considers a small class 
of MTs for some multicomponent evolution PDEs, but not the general multicomponent case.

We present a method to construct MTs for differential-difference equations in the general multicomponent case.
To clarify the main idea, we first consider the scalar case 
in Theorem~\ref{thscal} in Section~\ref{secgt}, 
and then the general multicomponent case in Theorem~\ref{tharb}.
Note that we give detailed proofs of these results.

A detailed exposition of the method is given in Section~\ref{secgt}.
The main ideas are briefly outlined in Section~\ref{secmi}.

To illustrate the method, in Section~\ref{secgt} we include a derivation 
of some MTs for the Volterra equation. 
These MTs (up to a change of variables) can be found also in~\cite{yam2006}.

Using the described method, 
in Sections~\ref{secbog},~\ref{sectoda},~\ref{secap} 
we construct a number of MTs for the Narita-Itoh-Bogoyavlensky, Toda lattices 
and Adler-Postnikov lattices from~\cite{adler-pos11}.
Some of the constructed MTs and modified equations seem to be new.

Some abbreviations, conventions, and notation used in the paper are presented in Section~\ref{subs-conv}.

\begin{remark}
\lb{idlr}

Let $\cS$ be the shift operator which replaces $n$ by $n+1$ and $u_l$ by $u_{l+1}$ 
in all considered functions. (A more detailed definition of $\cS$ is given in
Section~\ref{ddelr}.)

Let $\sm$ be a positive integer. 
Let $M=M(\ud,\la)$ and $\lt=\lt(\ud,\la)$ be $\sm\times\sm$ matrix-functions
depending on a finite number of the variables $u_l$ and a complex parameter~$\la$.
Suppose that the matrix $M$ is invertible, and the equation
\begin{equation}
\lb{ilr}
\pd_t(M)=\cS(\lt)M-M\lt
\end{equation}
holds as a consequence of~\eqref{isdde}. 
It is well known that such a pair $(M,\,\lt)$ is an analogue 
of a zero-curvature (Lax) representation for differential-difference equations.
Also, it is known that such $(M,\,\lt)$ often arise from Darboux transformations 
of PDEs (see, e.g.,~\cite{kmw}).

Motivated by these facts, following~\cite{kmw}, we say that the pair $(M,\,\lt)$ is 
a \emph{Darboux-Lax representation} (DLR) for equation~\eqref{isdde}.
This implies that the auxiliary linear system 
\begin{equation}
\lb{isyspsi}
\begin{aligned}
\cS(\Psi)&=M\Psi,\\
\pd_t(\Psi)&=\lt\Psi
\end{aligned}
\end{equation}
is compatible modulo~\eqref{isdde}. 
Here $\Psi=\Psi(n,t)$ is an invertible $\sm\times\sm$ matrix-function.
(A more precise definition of DLRs is given in Section~\ref{ddelr}.) 

Some authors say that a pair $(M,\,\lt)$ (with invertible $M$) 
is a DLR for~\eqref{isdde} 
if equation~\eqref{ilr} is equivalent to equation~\eqref{isdde}.
For our purposes, it is sufficient to assume that~\eqref{ilr} holds as a consequence of~\eqref{isdde}.

\end{remark}

\begin{remark}
\lb{yam-met}
A different method to construct MTs for differential-difference equations is described 
by Yamilov~\cite{yam94}.
Yamilov's method uses the following diagram, 
where $(F)$, $(G)$, $(T)$, $(H)$ are differential-difference equations 
and $m_1$, $m_2$, $m_3$, $m_4$ are  MTs.
\begin{equation}
\lb{yd}
\begin{CD}
(F) @<m_3<< (H)\\
@VVm_2V @VVm_4V\\
(T) @<m_1<< (G)
\end{CD}
\end{equation}
For given equations $(F)$, $(G)$, $(T)$ and MTs $m_1$, $m_2$ of order~$1$ 
satisfying certain conditions, 
the paper~\cite{yam94} constructs another equation $(H)$ and 
MTs $m_3$, $m_4$ of order~$1$ so that the diagram~\eqref{yd} is commutative.

Yamilov's method is very different from ours.
Indeed, we construct MTs of arbitrary order 
from certain Lie group actions associated with a given DLR.
Yamilov~\cite{yam94} constructs new MTs of order~$1$ from given MTs of order~$1$, 
not using any Lie group actions.
(Compositions of such MTs of order~$1$ can produce MTs of higher order~\cite{yam94}.)

It would be interesting to determine if Yamilov's MTs can be obtained 
from some DLRs by our method. (But we do not study Yamilov's MTs 
in the present paper.)

\end{remark}

\begin{remark}
\lb{remsok}
As has been said above, in this paper we study the following problem:
how to construct MTs~\eqref{iuvf} and equations~\eqref{ivdde} 
for a given equation~\eqref{isdde}?

Sokolov~\cite{sokolov88} studied the inverse problem:
how to construct MTs~\eqref{iuvf} and equations~\eqref{isdde} 
for a given equation~\eqref{ivdde}?
(More precisely, Sokolov~\cite{sokolov88} studied this for 
$(1+1)$-dimensional evolution PDEs and more general MTs 
which may change the space variable.)
According to~\cite{sokolov88}, the study of the inverse problem 
does not require Lax representations.
\end{remark}

\subsection{Abbreviations, conventions, and notation}
\lb{subs-conv}

The following abbreviations, conventions, and notation are used in the paper.

MT = Miura-type transformation, DLR = Darboux-Lax representation. 

The symbols $\zsp$ and $\zp$ denote the sets of positive and nonnegative 
integers respectively.

Each considered function 
is supposed to be analytic on its domain of definition.
(This includes meromorphic functions. By our convention, the poles of
a meromorphic function do not belong to its domain of definition, so 
a meromorphic function is analytic on its domain of definition.)

Unless otherwise specified, all scalar variables and functions are assumed to be 
$\fik$-valued.
The symbol~$n$ denotes an integer variable.

All considered Lie groups and Lie subgroups are supposed to be complex-analytic.

For every $\sm\in\zsp$, we denote by 
$\gl_\sm(\fik)$ the algebra of $\sm\times\sm$ matrices with entries from $\fik$ 
and by $\GL_\sm(\fik)$ the group of invertible $\sm\times\sm$ matrices with entries from $\fik$.

For $i\in\zsp$ and $k\in\zz$, the symbols $\cp_i$ denote complex constants, 
and $\cp_i^k=(\cp_i)^k$ is the $k$th power of~$\cp_i$.

Also, we use the notation $u^i_k=\cS^k(u^i)$ and $v^i_k=\cS^k(v^i)$, 
where $\cS^k$ is the $k$th power of the operator~$\cS$, 
and $u^i$, $v^i$ are dependent variables in differential-difference equations.

\subsection{The main ideas}
\lb{secmi}

In this subsection we outline the main ideas of our method to construct MTs
from DLRs.

Let $\sm\in\zsp$. 
Let $(M,\,\lt)$ be a $\sm\times\sm$ matrix DLR for equation~\eqref{isdde}, 
as described in Remark~\ref{idlr}.
So the matrix $M$ is invertible, and system~\eqref{isyspsi} is compatible modulo~\eqref{isdde}.

Our method to construct MTs is applicable to the case when the matrix 
$M=M(u,\la)$ depends only on $u$ and $\la$. 
(That is, $M$ does not depend on $u_l$ for $l\neq 0$.)
So we consider the case when $M$, $\lt$ are of the form 
\begin{equation}
\lb{imu}
M=M(u,\la),\qquad\quad
\lt=\lt(u_\ta,u_{\ta+1},\dots,u_\tb,\la)
\end{equation}
for some integers $\ta\le\tb$.

Without loss of generality, we can assume $a<0$.
(Indeed, $\lt(u_\ta,u_{\ta+1},\dots,u_\tb,\la)$ does not have to depend 
on~$u_\ta$ nontrivially, so we can always take $a<0$.)
Note that some DLRs of other types can be transformed to the form~\eqref{imu} 
by a change of variables, see Remark~\ref{rrelab} in Section~\ref{ddelr}.

\begin{remark}
\lb{reminv}
In some examples of DLRs, 
it may happen that $M(u,\la)$ is invertible for almost all (but not all) 
values of $u$ and $\la$.
The exceptional points $(u,\la)$ where $M(u,\la)$ is not invertible 
are excluded from further consideration.
\end{remark}

Consider the algebra $\fik[\la]$ of polynomials in~$\la$.
Fix $\cp\in\fik$ and $k\in\zsp$.
We denote by ${((\la-\cp)^k)\subset\fik[\la]}$ the ideal generated 
by the polynomial $(\la-\cp)^k\in\fik[\la]$.
Clearly, the quotient algebra $\fik[\la]/((\la-\cp)^k)$ is $k$-dimensional.

Denote by $\gq{\sm}{\cp}{k}$ the group 
of invertible $\sm\times\sm$ matrices with entries from $\fik[\la]/((\la-\cp)^k)$.
Equivalently, an element of $\gq{\sm}{\cp}{k}$ can be described as a sum of the form
\begin{equation}
\lb{iwww}
\sum_{q=0}^{k-1}(\la-\cp)^qW^q,\quad\qquad W^0\in\GL_\sm(\fik),\qquad
W^1,\dots,W^{k-1}\in\gl_\sm(\fik).
\end{equation}
This shows that $\gq{\sm}{\cp}{k}$ is a $k\sm^2$-dimensional Lie group.
When we multiply two elements of the form~\eqref{iwww} in the group~$\gq{\sm}{\cp}{k}$, 
we use the fact that in the algebra $\fik[\la]/((\la-\cp)^{k})$ 
one has $(\la-\cp)^{k}=0$.

For $i,j=1,\dots,\sm$ and $q=0,\dots,k-1$, 
we denote by $w^{q}_{ij}=(W^{q})_{ij}$ the entries of 
the matrix $W^{q}$. Then $w^{q}_{ij}$ can be regarded as coordinates 
on the Lie group $\gq{\sm}{\cp}{k}$. Set $G=\gq{\sm}{\cp}{k}$.

Now assume that $w^{q}_{ij}=w^{q}_{ij}(n,t)$ are functions of the variables $n$, $t$.
This means that $W^q$ are $\sm\times\sm$ matrix-functions of the variables $n$, $t$.
Substituting $\Psi=\sum_{q=0}^{k-1}(\la-\cp)^qW^q$ in~\eqref{isyspsi}, one obtains 
the system 
\begin{gather}
\lb{icsw}
\sum_{q=0}^{k-1}(\la-\cp)^q\cS\big(W^{q}\big)=
M(u,\la)\bigg(\sum_{q=0}^{k-1}(\la-\cp)^qW^{q}\bigg),
\qquad
(\la-\cp)^{k}=0,\\
\lb{idtw}
\sum_{q=0}^{k-1}(\la-\cp)^q\pd_t\big(W^{q}\big)=
\lt(u_\ta,\dots,u_\tb,\la)
\bigg(\sum_{q=0}^{k-1}(\la-\cp)^qW^{q}\bigg),
\qquad
(\la-\cp)^{k}=0.
\end{gather}
To compute the right-hand sides of~\eqref{icsw} and~\eqref{idtw}, 
we take the corresponding Taylor series with respect to $\la$ 
and truncate these series at the term $(\la-\cp)^{k-1}$.
(Recall that in the algebra $\fik[\la]/((\la-\cp)^{k})$ 
one has $(\la-\cp)^{k}=0$, which we use here.)

Since $w^{q}_{ij}=(W^{q})_{ij}$ are the entries of the matrix $W^{q}$, 
equation~\eqref{icsw} allows us to express $\cS(w^{q}_{ij})$ in terms 
of $u$ and $w^{\tq}_{\ti\tj}$ for $\tq=0,\dots,k-1,\ \ti,\tj=1,\dots,\sm$.
Similarly, equation~\eqref{idtw} allows us to express $\pd_t(w^{q}_{ij})$ in terms 
of $u_\ta,\dots,u_\tb$ and $w^{\tq}_{\ti\tj}$. That is, system~\eqref{icsw},~\eqref{idtw}
can be rewritten as 
\begin{equation}
\lb{wicsw}
\cS(w^{q}_{ij})=A^q_{ij}(u,w^{\tq}_{\ti\tj}),\qquad\quad
\pd_t(w^{q}_{ij})=B^q_{ij}(u_\ta,\dots,u_\tb,w^{\tq}_{\ti\tj})
\end{equation}
for some scalar functions $A^q_{ij}$, $B^q_{ij}$.

Recall that $w^{q}_{ij}$ are coordinates on the Lie group $G=\gq{\sm}{\cp}{k}$.
In what follows, by a function on~$G$ we mean an analytic function defined on an open dense subset of~$G$.

Let $z=z(w^{q}_{ij})$ be a function on $G$. 
Using formulas~\eqref{wicsw} and the property $\cS(u_l)=u_{l+1}$, 
one can compute $\pd_t(z)$ and $\cS^i(z)$ for any $i\in\zsp$.
Here and below $\cS^i$ denotes the $i$th power of the operator $\cS$.

To clarify the main idea of the method to construct MTs, suppose that $u$ is a scalar function.
In the scalar case,  
our goal is to find a function $z=z(w^{q}_{ij})$ such that there is $r\in\zsp$ 
satisfying the following conditions.
\begin{gather}
\lb{vcon1}
\begin{array}{c}
\text{$\cS^r(z)$ can be expressed in terms 
of $\,z,\ \cS(z),\ \cS^2(z),\,\dots,\ \cS^{r-1}(z),\ u\,$}\\
\text{so that the obtained expression
$\cS^r(z)=\fe\big(z,\cS(z),\cS^2(z),\dots,\cS^{r-1}(z),u\big)$}\\
\text{depends nontrivially on $u$.}
\end{array}\\
\lb{vcon2}
\begin{array}{c}
\text{$\pd_t(z)$ can be expressed in terms of 
$\,z,\ \cS(z),\ \cS^2(z),\,\dots,\ \cS^{r-1}(z),\ u_\ta,\,\dots,\ u_\tb$,}\\
\text{so $\pd_t(z)=Q\big(z,\cS(z),\cS^2(z),\dots,\cS^{r-1}(z),u_\ta,\dots,u_\tb\big)$ 
for some function $Q$.}
\end{array}
\end{gather}

Using a function $z$ satisfying~\eqref{vcon1},~\eqref{vcon2}, 
one obtains an MT for equation~\eqref{isdde} as follows.
Set $z_0=z$ and $z_i=\cS^i(z)$ for $i=1,\dots,r$. 
Then the above formulas 
\begin{gather*}
\cS^r(z)=\fe\big(z,\cS(z),\cS^2(z),\dots,\cS^{r-1}(z),u\big),\\
\pd_t(z)=Q\big(z,\cS(z),\cS^2(z),\dots,\cS^{r-1}(z),u_\ta,\dots,u_\tb\big)
\end{gather*}
become
\begin{gather}
\lb{vrfe}
z_r=\fe(z_0,z_1,z_2,\dots,z_{r-1},u),\\
\lb{dtvq}
\pd_t(z)=Q\big(z_0,z_1,z_2,\dots,z_{r-1},u_\ta,\dots,u_\tb\big).
\end{gather}
Since, according to~\eqref{vcon1}, the function $\fe$ depends nontrivially on $u$, 
locally from equation~\eqref{vrfe} one can express $u$ in terms of $z_0,z_1,\dots,z_r$
\begin{equation}
\lb{iufz}
u=\mf(z_0,z_1,\dots,z_{r}).
\end{equation}
We introduce new variables $\fv_l$ for $l\in\mathbb{Z}$.
Using the function $\mf(z_0,z_1,\dots,z_{r})$ from~\eqref{iufz}, 
for each $j\in\mathbb{Z}$ one can consider the function 
$\mf(\fv_{j},\fv_{j+1},\dots,\fv_{j+r})$.
Let $P(\fv_{\ta},\dots,\fv_{{\tb}+r})$ be the function 
obtained from $Q(z_0,\dots,z_{r-1},u_\ta,\dots,u_\tb)$
by replacing $z_i$ with $\fv_i$ and $u_j$ with $\mf(\fv_{j},\fv_{j+1},\dots,\fv_{j+r})$.
That is,
\begin{equation}
\notag
P(\fv_{\ta},\dots,\fv_{{\tb}+r})
=Q\big(\fv_0,\dots,\fv_{r-1},\mf(\fv_{\ta},\dots,\fv_{\ta+r}),\dots,
\mf(\fv_{\tb},\dots,\fv_{\tb+r})\big).
\end{equation}
We introduce the formula
\begin{equation}
\lb{ivtp}
v_t=P(\fv_{\ta},\dots,\fv_{{\tb}+r})
\end{equation}
and regard~\eqref{ivtp} as a differential-difference equation for~$\fv$.
Then the formula 
\begin{equation}
\lb{iumffv}
u=\mf(\fv_0,\fv_1,\dots,\fv_{r})
\end{equation}
determines an MT from equation~\eqref{ivtp} to equation~\eqref{isdde}. 
As usual, we can use the identification $v_0=v$. 
Then \eqref{iumffv} becomes $u=\mf(\fv,\fv_1,\dots,\fv_{r})$.

So we have shown that a function $z=z(w^{q}_{ij})$ satisfying~\eqref{vcon1},~\eqref{vcon2} gives an MT. 
Now let us explain how to construct such~$z$.
It turns out that one can construct such functions~$z$ as invariants 
with respect to the left and right actions of certain subgroups of~$G$.
(Subgroups of~$G$ act on~$G$ by left and right   
multiplication, which induces actions on the algebra of functions on~$G$.)

In Section~\ref{sgalr}, for a given DLR~\eqref{imu} we define a sequence of groups 
\begin{equation}
\notag
\hg{0}\subset\hg{1}\subset\hg{2}\subset\dots
\end{equation}
For each $p\in\zp$, the group $\hg{p}$ consists of certain invertible 
$\sm\times\sm$ matrix-functions of~$\la$.
The groups $\hg{p}$ are defined by induction on $p$ as follows.

We set $\hg{0}=\iel$, where $\iel$ is the identity matrix of size~$\sm$.
The group $\hg{1}$ is generated by the matrix-functions 
$M({\tilde u},\la)\cdot M(u,\la)^{-1}$ for all values of $u$, ${\tilde u}$.

For each $p\in\zsp$, the group $\hg{p+1}$ is generated by the elements 
$h_p\in\hg{p}$ and the elements 
$$
M(u,\la)\cdot h_p\cdot M(u,\la)^{-1}
$$ 
for all $h_p\in\hg{p}$ and all values of $u$.

Since the elements of $\hg{p}$ are invertible $\sm\times\sm$ matrix-functions of~$\la$, 
the group $\hg{p}$ acts on the group $G=\gq{\sm}{\cp}{k}$ by left multiplication.
(Performing this multiplication, we work modulo the relation ${(\la-\cp)^{k}=0}$, 
as has been said above.) A function $\mathbf{f}$ on~$G$ is called 
\emph{$\hg{p}$-left-invariant} if $\mathbf{f}$ is invariant with respect 
to this left action of~$\hg{p}$.

Let $\mathcal{H}\subset G$ be a subgroup of $G$. 
A function $\mathbf{f}$ on~$G$ is called 
\emph{$\mathcal{H}$-right-invariant} if $\mathbf{f}$ is invariant with respect 
to the action of~$\mathcal{H}$ on~$G$ by right multiplication.
(See Section~\ref{sgalr} for a detailed definition of left-invariant 
and right-invariant functions.)

For an arbitrary function $z=z(w^{q}_{ij})$ on $G$, 
the functions $\cS^\ga(z)$ for $\ga\in\zsp$ may depend on $w^{q}_{ij}$ and $u_l$.
In Section~\ref{secgt} we prove the following statements. 

If a function~$z$ on~$G$ is $\hg{p}$-left-invariant for some $p\in\zsp$ then 
the functions 
\begin{equation}
\lb{icscs}
\cS(z),\quad \cS^2(z),\quad\dots,\quad\cS^{p}(z)
\end{equation}
do not depend on $u_l$ for any $l\in\zz$. 
So $\cS(z)$, $\cS^2(z),\,\dots,\,\cS^{p}(z)$ depend only on $w^{q}_{ij}$ 
and can be regarded as functions on~$G$. 
This implies that $\cS^{p+1}(z)$ may depend only on $w^{q}_{ij}$ and $u=u_0$. 

Furthermore, for any closed connected Lie subgroup $H\subset G$, 
if a function~$z$ on~$G$ is $H$-right-invariant and 
is $\hg{p}$-left-invariant for some $p\in\zsp$ 
then the functions~\eqref{icscs} are $H$-right-invariant as well.

These statements (along with some other considerations) allow us to prove the following.
Suppose that a function $z$ on~$G$ satisfies the following conditions:
\begin{gather}
\lb{nvc1}
\begin{array}{c}
\text{there is a closed connected Lie subgroup $H\subset G$ of codimension $r>0$}\\ 
\text{such that $z$ is $H$-right-invariant and $\hg{r-1}$-left-invariant,}
\end{array}\\
\lb{nvc2}
\begin{array}{c}
\text{the differentials of the functions 
$z$, $\cS(z)$, $\cS^2(z),\,\dots,\,\cS^{r-1}(z)$}\\ 
\text{are linearly independent on an open dense subset of $G$,}
\end{array}\\
\lb{nvc3}
\begin{array}{c}
\text{the function $\cS^{r}(z)$ depends nontrivially on $u$.}
\end{array}
\end{gather}
Then $z$ obeys~\eqref{vcon1},~\eqref{vcon2}.

So~\eqref{vcon1},~\eqref{vcon2} follow from~\eqref{nvc1},~\eqref{nvc2},~\eqref{nvc3}.
As has been shown above, a function~$z$ obeying~\eqref{vcon1},~\eqref{vcon2} gives an MT.

Note that \eqref{nvc2},~\eqref{nvc3} can be regarded as non-degeneracy conditions.
In all examples known to us, if a nonconstant function~$z$ satisfies~\eqref{nvc1} 
then \eqref{nvc2},~\eqref{nvc3} are also satisfied.
Therefore, to construct an MT, 
one needs to find a nonconstant function~$z$ obeying~\eqref{nvc1}.
This can often be done as follows.

Note that $H$-right-invariant functions on~$G$ 
can be identified with functions on the quotient manifold~$G/H$.
For any $p\in\zp$, the above-mentioned left action of~$\hg{p}$ on $G$ 
induces a left action of~$\hg{p}$ on~$G/H$.

Therefore, we need to choose a positive integer~$r$ 
and a closed connected Lie subgroup $H\subset G$ of codimension $r>0$
such that the left action of~$\hg{r-1}$ on~$G/H$ possesses a nonconstant 
invariant~$z$. (That is, $z$ is a nonconstant $\hg{r-1}$-left-invariant function on~$G/H$.)
Then $z$ can be viewed as an $H$-right-invariant function on~$G$ and obeys~\eqref{nvc1}.

A detailed description of this theory is given in Section~\ref{secgt}.
Actually, in Section~\ref{secgt} we consider a more general Lie group $G$ 
which is equal to the product of several groups of the form $\gq{\sm}{\cp}{k}$.
(See formula~\eqref{ggr}.) Because of this, some notation in Section~\ref{secgt} 
is different from the notation in the present section.

In the above discussion we have assumed that $u$ is a scalar function.
In Section~\ref{secgt} we develop this theory in the case when 
$u=\big(u^1(n,t),\dots,u^{\diu}(n,t)\big)$ is 
an $\diu$-component vector-function for arbitrary $\diu\ge 1$.
To construct an MT in the case $\diu>1$, 
we use several invariants with respect to 
the left actions of the groups $\hg{p}$ and the right action  
of a Lie subgroup $H\subset G$,
which must satisfy certain conditions (see Theorem~\ref{tharb} for details).

\begin{remark}
\lb{conjug}
Let $\mathbf{g}(\la)$ be a $\GL_\sm(\fik)$-valued function of~$\la$.
Using the DLR~\eqref{imu}, consider the matrix-functions
\begin{equation}
\lb{tmtu}
\tilde M=\mathbf{g}(\la)\cdot M\cdot\mathbf{g}(\la)^{-1},\qquad\quad
\tilde\lt=\mathbf{g}(\la)\cdot\lt\cdot\mathbf{g}(\la)^{-1}.
\end{equation}
Then~\eqref{tmtu} is a DLR as well, which is said to be
\emph{gauge equivalent} to~\eqref{imu}.

Replacing~\eqref{imu} by~\eqref{tmtu}, one can sometimes simplify the form  
of the groups~$\hg{p}$. Examples of this procedure are presented in Section~\ref{secbog}.
\end{remark}

\begin{remark}
\lb{opendense}
As has been discussed above, 
if a function~$z$ on~$G$ is $\hg{p}$-left-invariant for some $p\in\zsp$ then 
the functions~\eqref{icscs} do not depend on $u_l$ for any $l\in\zz$.

Recall that by a function on~$G$ we mean an analytic function defined on an open dense subset of~$G$. So $z$ is defined on an open dense subset $\mathbb{U}\subset G$.
Then the definition of the operator $\cS$ implies that for each $l=1,\dots,p$ 
the function $\cS^l(z)$ is defined on an open dense subset $\mathbb{U}^l\subset G$, 
and it may happen that $\mathbb{U}^l\neq\mathbb{U}$.

Since $\mathbb{U}$, $\mathbb{U}^1,\dots,\mathbb{U}^p$ are open dense subsets of~$G$, 
the intersection $\mathbb{U}\cap\bigcap_{l=1}^p\mathbb{U}^l$ is also open and dense 
in~$G$. Then we can regard~\eqref{icscs} as functions on 
$\mathbb{U}\cap\bigcap_{l=1}^p\mathbb{U}^l$.

\end{remark}

\section{General theory}
\lb{secgt}

\subsection{Differential-difference equations and Darboux-Lax representations}
\lb{ddelr}

Fix $\diu\in\zsp$ and $\ua,\ub\in\zz$ such that $\ua\le\ub$.
Consider a differential-difference equation of the form
\begin{equation}
\lb{sdde}
u_t=\ff(u_\ua,u_{\ua+1},\dots,u_\ub)
\end{equation}
for an $\diu$-component vector-function $u=\big(u^1(n,t),\dots,u^{\diu}(n,t)\big)$ of
an integer variable~$n$ and a real or complex variable~$t$.
We use the standard notation $u_t=\pd_t(u)$ and $u_l=u(n+l,t)$ for $l\in\zz$.
In particular, $u_0=u$.

In other words, equation \eqref{sdde} encodes an infinite sequence of differential equations
$$
\pd_t\big(u(n,t)\big)=\ff\big(u(n+\ua,t),u(n+\ua+1,t),\dots,u(n+\ub,t)\big),\qquad\quad n\in\zz.
$$
Here $\ff$ is also an $\diu$-component vector-function $\ff=(\ff^1,\dots,\ff^{\diu})$.
So in components equation~\eqref{sdde} reads
\begin{equation}
\lb{msdde}
u_t^i=\ff^i(u_\ua^\ga,u_{\ua+1}^\ga,\dots,u_\ub^\ga),\qquad\quad
i=1,\dots,\diu.
\end{equation}

As usual in the formal theory of differential-difference equations,
one regards 
$$
u_l=(u^1_l,\dots,u^\diu_l),\qquad\quad l\in\zz,
$$ 
as independent quantities (which are sometimes 
called \emph{dynamical variables} in the literature).
When we write $f=f(\ud)$, we mean that $f$ may depend on any finite 
number of the variables $u_l^{\ga}$ for $l\in\zz$ and $\ga=1,\dots,\diu$.

When we write $f=f(u_p,\dots,u_q)$ or $f=f(u_p,u_{p+1},\dots,u_q)$ 
for some integers $p\le q$,
we mean that $f$ may depend on $u_l^{\ga}$ for $l=p,\dots,q$ 
and $\ga=1,\dots,\diu$.

Let $\cS$ be the \emph{shift operator} with respect to the integer variable $n$.
That is, for any function $g=g(n,t)$ one has $\cS(g)=g(n+1,t)$.

Since $u_l$ corresponds to $u(n+l,t)$, the operator $\cS$ 
and its powers $\cS^k$ for $k\in\zz$ act on functions of $u_l$ as follows
\begin{equation}
\lb{csuf}
\cS(u_l)=u_{l+1},\qquad\cS^k(u_l)=u_{l+k},\qquad
\cS^k\big(f(u_l,\dots)\big)=f(\cS^k(u_{l}),\dots).
\end{equation}
That is, applying $\cS^k$ to a function $f=f(\ud)$, 
we replace $u_l^{\ga}$ by $u_{l+k}^{\ga}$ in $f$ for all $l,\ga$.

The \emph{total derivative operator} $D_t$ corresponding to~\eqref{sdde} 
is given by the formula 
\begin{equation}
\lb{dtfu}
D_t(f)=\sum_{l,\ga}\cS^l(\ff^\ga)\cdot\frac{\pd f}{\pd u_l^\ga}.
\end{equation}

Let $\sm\in\zsp$. 
Let $M=M(\ud,\la)$ and $\lt=\lt(\ud,\la)$ be $\sm\times\sm$ matrix-functions
depending on the variables $u_l$ and a complex parameter $\la$.
Suppose that $M$ is invertible and 
\begin{equation}
\lb{lr}
D_t(M)=\cS(\lt)M-M\lt,
\end{equation}
where $D_t$ is given by~\eqref{dtfu}.
Then the pair $(M,\,\lt)$ is called a \emph{\textup{(}matrix\textup{)} Darboux-Lax representation} (DLR) for equation~\eqref{sdde}.
This implies that the auxiliary linear system 
\begin{equation}
\lb{syspsi}
\begin{aligned}
\cS(\Psi)&=M\Psi,\\
\pd_t(\Psi)&=\lt\Psi
\end{aligned}
\end{equation}
is compatible modulo~\eqref{sdde}. 
Here $\Psi=\Psi(n,t)$ is an invertible $\sm\times\sm$ matrix-function.

Our method to construct MTs is applicable to the case when 
\begin{equation}
\notag
M=M(u,\la)=M(u^1,\dots,u^{\diu},\la)
\end{equation}
depends only on 
$u=(u^1,\dots,u^{\diu})$ and $\la$. 
(That is, $M$ does not depend on $u_l$ for $l\neq 0$.)

So we will mainly consider the case when $M$, $\lt$ are of the form 
\begin{equation}
\lb{mu}
M=M(u,\la),\qquad\quad
\lt=\lt(u_\ta,u_{\ta+1},\dots,u_\tb,\la)
\end{equation}
for some integers $\ta\le\tb$.
Without loss of generality, we can assume $a<0$.

We will use also Remark~\ref{reminv} on invertibility of the matrix $M(u,\la)$.

\begin{remark}
\lb{rrelab}
If $M=M(u^1_{l_1},\dots,u^{\diu}_{l_\diu},\la)$ depends on $\la$ and 
$u^1_{l_1},\dots,u^{\diu}_{l_\diu}$ for some fixed integers $l_1,\dots,l_\diu$, 
then one can relabel 
$$
u^1:=u^1_{l_1},\quad\dots\quad,\quad u^{\diu}:=u^{\diu}_{l_\diu},
$$ 
which reduces $M(u^1_{l_1},\dots,u^{\diu}_{l_\diu},\la)$ to $M(u^1,\dots,u^{\diu},\la)$. 
An example of such relabelling is used in Section~\ref{sectoda}.
\end{remark}

Similarly to~\eqref{sdde}, consider a differential-difference equation
\begin{equation}
\lb{vdde}
v_t=\tff(v_\tua,v_{\tua+1},\dots,v_\tub)
\end{equation}
for an $\diu$-component vector-function $v=\big(v^1(n,t),\dots,v^{\diu}(n,t)\big)$.
Here $\tua,\tub\in\zz$, $\tua\le\tub$, and $v_l=v(n+l,t)$ for $l\in\zz$.
In particular, $v_0=v$.

Similarly to~\eqref{csuf},~\eqref{dtfu}, the operators $\cS$ and $D_t$ act on functions 
of the variables $v_l=(v^1_l,\dots,v^\diu_l)$ as follows
\begin{gather*}
\cS(v_l)=v_{l+1},\qquad\cS^k(v_l)=v_{l+k},\qquad
\cS^k\big(f(v_l,\dots)\big)=f(\cS^k(v_{l}),\dots),\qquad k\in\zz,\\
D_t\big(f(v_l,\dots)\big)=\sum_{l,\ga}\cS^l(\tff^\ga)\cdot\frac{\pd f}{\pd v_l^\ga},
\end{gather*}
where $\tff^\ga$ are the components of the vector-function 
$\tff=(\tff^1,\dots,\tff^\diu)$ from~\eqref{vdde}.

\begin{definition}
\lb{defmtt}
A \emph{Miura-type transformation} (MT) from equation~\eqref{vdde} to equation~\eqref{sdde} 
is determined by an expression of the form
\begin{equation}
\lb{uvf}
u=\mf(v_l,\dots)
\end{equation}
(where $\mf$ may depend on any finite number of the variables 
$v_l=(v^1_l,\dots,v^\diu_l)$, $l\in\zz$,)
such that if $v$ satisfies~\eqref{vdde} then $u$ given by~\eqref{uvf} satisfies~\eqref{sdde}.

More precisely, this means the following. 
In components formula~\eqref{uvf} reads 
\begin{equation}
\lb{uimfi}
u^i=\mf^i(v^\ga_l,\dots),\qquad\quad i=1,\dots,\diu,
\end{equation}
where $\mf^i$ are the components of the vector-function $\mf=(\mf^1,\dots,\mf^{\diu})$.
If we substitute the right-hand side of~\eqref{uimfi} in place of $u^i$ in~\eqref{msdde}, we get 
$$
D_t\big(\mf^i(v^\ga_l,\dots)\big)=
\ff^i\big(\cS^\ua(\mf^\ga),\cS^{\ua+1}(\mf^\ga),\dots,\cS^\ub(\mf^\ga)\big),\qquad\quad
i=1,\dots,\diu,
$$
which must be an identity in the variables $v^\ga_l$.
\end{definition}

\begin{remark}
Consider an MT $u=\mf(v_l,\dots)$ from equation~\eqref{vdde} to 
equation~\eqref{sdde}.
Suppose that equation~\eqref{sdde} possesses a Darboux-Lax 
representation (DLR) of the form~\eqref{mu}
such that the equation $D_t(M)=\cS(\lt)M-M\lt$ is equivalent to~\eqref{sdde}.

Let $\hat M=\hat M(v_l,\dots,\la)$ and 
$\hat\lt=\hat\lt(v_l,\dots,\la)$ be the matrix-functions obtained 
from~\eqref{mu} by substituting $\mf(v_l,\dots)$ in place of $u$.
So $\hat M=M\big(\mf(v_l,\dots),\la\big)$, 
and $\hat\lt$ is obtained from $\lt(u_\ta,\dots,u_\tb,\la)$ 
by substituting 
$\cS^i\big(\mf(v_l,\dots)\big)$ in place of $u_i$ for all $i\in\zz$.

Since $(M,\,\lt)$ is a DLR for~\eqref{sdde} and 
$u=\mf(v_l,\dots)$ is an MT from~\eqref{vdde} to~\eqref{sdde}, 
the obtained matrices $\hat M$, $\hat\lt$ form a DLR for equation~\eqref{vdde}
in the sense that the equation $D_t(\hat M)=\cS(\hat\lt)\hat M-\hat M\hat\lt$ 
is equivalent to a consequence of~\eqref{vdde}.

Then for any invertible $\sm\times\sm$ matrix  
$\mathbf{g}=\mathbf{g}(v_l,\dots,\la)$ the matrices
\begin{equation}
\lb{ggtmtu}
\check M=\cS(\mathbf{g})\cdot \hat M\cdot\mathbf{g}^{-1},\qquad\quad
\check\lt=D_t(\mathbf{g})\cdot\mathbf{g}^{-1}+
\mathbf{g}\cdot\hat\lt\cdot\mathbf{g}^{-1}
\end{equation}
form a DLR for~\eqref{vdde} as well.
The DLR $(\check M,\,\check\lt)$ is \emph{gauge equivalent}  
to the DLR $(\hat M,\,\hat\lt)$ 
with respect to the \emph{gauge transformation}~$\mathbf{g}$.

Very often, one can find a matrix $\mathbf{g}=\mathbf{g}(v_l,\dots,\la)$ 
such that the equation $D_t(\check M)=\cS(\check\lt)\check M-\check M\check\lt$ 
with $\check M$, $\check\lt$ given by~\eqref{ggtmtu} is equivalent to~\eqref{vdde}.
(So $D_t(\check M)=\cS(\check\lt)\check M-\check M\check\lt$ 
is equivalent to equation~\eqref{vdde} itself, while 
$D_t(\hat M)=\cS(\hat\lt)\hat M-\hat M\hat\lt$ 
is equivalent to a consequence of~\eqref{vdde}.)

For instance, consider the MT $u=v v_1$  
from the modified Volterra equation $v_t=v^2 (v_1-v_{-1})$
to the Volterra equation $u_t=u(u_1-u_{-1})$.
As we discuss in Example~\ref{exlrv} below,
the matrices~\eqref{lrvol} form a DLR for the Volterra equation.
Substituting $v v_1$ in place of $u$ in~\eqref{lrvol}, one obtains the matrices
\begin{equation}
\lb{hatlrvol}
\hat M=
\begin{pmatrix}
 0 & v v_1\\
 -1 & \lambda
 \end{pmatrix},
\qquad\quad
\hat\lt=
\begin{pmatrix}
 v v_1 & \lambda v_{-1}v \\
 -\lambda & \lambda^2+v_{-1}v 
\end{pmatrix}.
\end{equation}
The equation $D_t(\hat M)=\cS(\hat\lt)\hat M-\hat M\hat\lt$ 
is equivalent to $(v v_1)_t=vv_2(v_1)^2-v_1v_{-1}v^2$, 
which is a consequence of $v_t=v^2 (v_1-v_{-1})$.
 
Consider the gauge transformation 
$\mathbf{g}=\begin{pmatrix}
 1/v & 0 \\
 0 & 1 
\end{pmatrix}$. 
Computing \eqref{ggtmtu} for \eqref{hatlrvol}, we get
 \begin{equation}
\lb{chlrvol}
\check M=
\begin{pmatrix}
 0 & v\\
 -v & \lambda
 \end{pmatrix},
\qquad\quad
\check\lt=
\begin{pmatrix}
 v v_{-1} & \lambda v_{-1} \\
 -\lambda v & \lambda^2+vv_{-1} 
\end{pmatrix}.
\end{equation}
The equation $D_t(\check M)=\cS(\check\lt)\check M-\check M\check\lt$ 
is equivalent to $v_t=v^2 (v_1-v_{-1})$.
The DLR~\eqref{chlrvol} is well known.
\end{remark}

\subsection{Some algebraic and geometric structures 
associated with Darboux-Lax representations}
\lb{sgalr}

Fix $\sm\in\zsp$. 
Consider a DLR of the form~\eqref{mu}, 
where $M$, $\lt$ are $\sm\times\sm$ matrix-functions satisfying~\eqref{lr}, 
and $M$ is invertible.

Let $\gla$ be the group of \mbox{$\GL_\sm(\fik)$-valued} functions of~$\la$.
Since $M(u,\la)$ in~\eqref{mu} is supposed to be invertible, 
for every fixed value of~$u$ we have $M(u,\la)\in\gla$.

\begin{remark}
\lb{rgla}
The group $\gla$ consists of \mbox{$\GL_\sm(\fik)$-valued} functions of~$\la$ 
defined on some connected open subset of~$\fik$. 
Essentially, the only requirement is that $M(u,\la)\in\gla$ for every fixed value of~$u$, 
where $u$ runs through some connected open subset of~$\fik^\diu$.
To simplify notation, we do not mention these open subsets explicitly.

As has been said in Section~\ref{subs-conv}, each considered function 
is supposed to be analytic on its domain of definition.
\end{remark}

We define a sequence of subgroups 
$$
\hg{0}\subset\hg{1}\subset\hg{2}\subset\dots\subset\gla,
$$
associated with the DLR~\eqref{mu} as follows.

We set $\hg{0}=\iel$, where $\iel\in\gla$ is the identity element.
The subgroup $\hg{1}\subset\gla$ is generated by the elements 
\begin{equation}
\lb{mumu}
M({\tilde u},\la)\cdot M(u,\la)^{-1}\in\gla
\end{equation}
for all values of $u$, ${\tilde u}$.

Now we define $\hg{k}\subset\gla$ by induction on $k\in\zsp$ as follows.
For each $k\in\zsp$, the subgroup $\hg{k+1}\subset\gla$ is generated by the elements 
$h_k\in\hg{k}$ and the elements 
\begin{equation}
\lb{mhm}
M(u,\la)\cdot h_k\cdot M(u,\la)^{-1}\in\gla
\end{equation}
for all $h_k\in\hg{k}$ and all values of $u$.

\begin{example}
\lb{exlrv}
For the Volterra equation 
\begin{equation}
\lb{volt}
u_t=u(u_1-u_{-1}), 
\end{equation}
one has $\diu=1$, and we can take 
\begin{equation}
\lb{lrvol}
M(u,\la)=
\begin{pmatrix}
 0 & u\\
 -1 & \lambda
 \end{pmatrix},
\qquad\quad
\lt(u_{-1},u,\la)=
\begin{pmatrix}
 u & \lambda u_{-1} \\
 -\lambda & \lambda^2+u_{-1} 
\end{pmatrix}.
\end{equation}
Note that $M(u,\la)=
\begin{pmatrix}
 0 & u\\
 -1 & \lambda
 \end{pmatrix}$
is not invertible for $u=0$. 
In agreement with Remark~\ref{reminv}, we assume $u\neq 0$.

In this example we have $\sm=2$, 
so $\gla$ consists of $\GL_2(\fik)$-valued functions of~$\la$.
The subgroup $\hg{1}\subset\gla$ is generated by the elements 
\begin{equation}
\notag
M({\tilde u},\la)\cdot M(u,\la)^{-1}=
\begin{pmatrix}
 0 & \tilde{u}\\
 -1 & \lambda\\
 \end{pmatrix}
\cdot
\begin{pmatrix}
\frac{\lambda}{u} & -1\\
\frac{1}{u} & 0\\
\end{pmatrix}=
\begin{pmatrix}
 \frac{\tilde u}{u} & 0 \\
 0 & 1 
\end{pmatrix}.
\end{equation}
Hence $\hg{1}$ consists of the constant matrix-functions 
$\begin{pmatrix}
 a_1 & 0 \\
 0 & 1 
\end{pmatrix}$,  
where $a_1\in\fik$ is an arbitrary nonzero constant.

The subgroup $\hg{2}\subset\gla$ is generated by the elements 
$h_1=\begin{pmatrix}
 a_1 & 0 \\
 0 & 1 
\end{pmatrix}\in\hg{1}$ and the elements 
\begin{equation}
\notag
M(u,\la)\cdot h_1\cdot M(u,\la)^{-1}=
\begin{pmatrix}
 1 & 0 \\
 \frac{(1-a_1)\lambda}{u} & a_1
\end{pmatrix}
\end{equation}
for all nonzero $a_1,u\in\fik$.
Therefore, $\hg{2}$ consists of the matrix-functions 
$\begin{pmatrix}
 a_1 & 0 \\
 a_3\lambda & a_2 
\end{pmatrix}$
for all $a_1,a_2,a_3\in\fik$, $a_1\neq 0$, $a_2\neq 0$.
\end{example}

Consider the algebra $\fik[\la]$ of polynomials in~$\la$.
For every $\cp\in\fik$ and $k\in\zsp$, 
we denote by $((\la-\cp)^k)\subset\fik[\la]$ the ideal generated 
by the polynomial $(\la-\cp)^k\in\fik[\la]$.
Clearly, the quotient algebra $\fik[\la]/((\la-\cp)^k)$ is $k$-dimensional.

Denote by $\gq{\sm}{\cp}{k}$ the group 
of invertible $\sm\times\sm$ matrices with entries from $\fik[\la]/((\la-\cp)^k)$.
Equivalently, an element of $\gq{\sm}{\cp}{k}$ can be described as a sum of the form
\begin{equation}
\lb{www}
\sum_{q=0}^{k-1}(\la-\cp)^qW^q,\quad\qquad W^0\in\GL_\sm(\fik),\qquad
W^1,\dots,W^{k-1}\in\gl_\sm(\fik).
\end{equation}
This shows that $\gq{\sm}{\cp}{k}$ is a $k\sm^2$-dimensional Lie group.
When we multiply two elements of the form~\eqref{www} in the group~$\gq{\sm}{\cp}{k}$, 
we use the fact that in the algebra $\fik[\la]/((\la-\cp)^{k})$ 
one has $(\la-\cp)^{k}=0$.

Fix 
$$
\og\in\zsp,\quad\qquad \cp_1,\dots,\cp_\og\in\fik,
\quad\qquad
k_1,\dots,k_\og\in\zsp.
$$ 

\begin{remark}
\lb{cpbel}
According to Remark~\ref{rgla}, 
the group $\gla$ consists of \mbox{$\GL_\sm(\fik)$-valued} functions of~$\la$ 
defined on some connected open subset of~$\fik$. We assume that 
$\cp_1,\dots,\cp_\og$ belong to this open subset.
\end{remark}

Consider the following Lie group
\begin{equation}
\lb{ggr}
G=\gq{\sm}{\cp_1}{k_1}\times\dots\times\gq{\sm}{\cp_\og}{k_\og}.
\end{equation}
An element $g\in G$ is an $\og$-tuple $g=(g^1,\dots,g^\og)$, 
where $g^p\in\gq{\sm}{\cp_p}{k_p}$, $p=1,\dots,\og$, can be described as a sum of the form
\begin{equation}
\notag
g^p=\sum_{q=0}^{k_p-1}(\la-\cp_p)^qW^{p,q},\quad\qquad W^{p,0}\in\GL_\sm(\fik),\qquad
W^{p,1},\dots,W^{p,k_p-1}\in\gl_\sm(\fik).
\end{equation}
For $i,j=1,\dots,\sm$, we denote by $w^{pq}_{ij}=(W^{p,q})_{ij}$ the entries of 
the matrix $W^{p,q}$. Then $w^{pq}_{ij}$ can be regarded as coordinates 
on the Lie group $G$.

Consider functions of the form $f(w^{pq}_{ij},\ud)$ which may depend 
on the coordinates $w^{pq}_{ij}$ and a finite number of the variables 
$u_l=(u^1_l,\dots,u^\diu_l)$, $l\in\zz$. 
Using the DLR~\eqref{mu}, 
we extend the operators $\cS$, $D_t$ to such functions as follows.
 
To define $\cS(w^{pq}_{ij})$ and $D_t(w^{pq}_{ij})$, we use the formulas
\begin{gather}
\lb{csw}
\sum_{q=0}^{k_p-1}(\la-\cp_p)^q\cS\big(W^{p,q}\big)=
M(u,\la)\bigg(\sum_{q=0}^{k_p-1}(\la-\cp_p)^qW^{p,q}\bigg),
\qquad
(\la-\cp_p)^{k_p}=0,\\
\lb{dtw}
\sum_{q=0}^{k_p-1}(\la-\cp_p)^qD_t\big(W^{p,q}\big)=
\lt(u_\ta,\dots,u_\tb,\la)
\bigg(\sum_{q=0}^{k_p-1}(\la-\cp_p)^qW^{p,q}\bigg),
\qquad
(\la-\cp_p)^{k_p}=0.
\end{gather}
To compute the right-hand sides of~\eqref{csw} and~\eqref{dtw}, 
we take the corresponding Taylor series with respect to $\la$ 
and truncate these series at the term $(\la-\cp_p)^{k_p-1}$.
(Recall that in the algebra $\fik[\la]/((\la-\cp_p)^{k_p})$ 
one has $(\la-\cp_p)^{k_p}=0$, which we use here.)

Since $\cS(w^{pq}_{ij})$ and $D_t(w^{pq}_{ij})$ appear on the left-hand sides 
of~\eqref{csw} and~\eqref{dtw} respectively, formulas~\eqref{csw},~\eqref{dtw} define 
$\cS(w^{pq}_{ij})$ and $D_t(w^{pq}_{ij})$. 
Now for a function $f(w^{pq}_{ij},\ud)$ we set
\begin{gather*}
\cS\big(f(w^{pq}_{ij},\ud)\big)=f(\cS(w^{pq}_{ij}),\cS(u_l),\dots),\\
D_t\big(f(w^{pq}_{ij},\ud)\big)=\sum_{p,q,i,j}\frac{\pd f}{\pd w^{pq}_{ij}}
D_t(w^{pq}_{ij})+\sum_{l,\ga}\frac{\pd f}{\pd u_l^\ga}D_t(u_l^\ga),
\end{gather*}
where we use the fact that 
$\cS(w^{pq}_{ij})$, $\cS(u_l)$, $D_t(w^{pq}_{ij})$, $D_t(u_l^\ga)$ 
have already been defined.

The presented definition of the operator $\cS$ implies that the inverse~$\cS^{-1}$
is given by
\begin{gather}
\lb{s1w}
\sum_{q=0}^{k_p-1}(\la-\cp_p)^q\cS^{-1}\big(W^{p,q}\big)=
M(u_{-1},\la)^{-1}\bigg(\sum_{q=0}^{k_p-1}(\la-\cp_p)^qW^{p,q}\bigg),
\qquad
(\la-\cp_p)^{k_p}=0,\\
\notag
\cS^{-1}\big(f(w^{pq}_{ij},\ud)\big)=f(\cS^{-1}(w^{pq}_{ij}),\cS^{-1}(u_l),\dots),
\end{gather}
where $\cS^{-1}(u_l)=u_{l-1}$ and 
$\cS^{-1}(w^{pq}_{ij})$ is determined by~\eqref{s1w}.

\begin{lemma}
\lb{sdts}
One has $\cS\circ D_t=D_t\circ\cS$.
\end{lemma}
\begin{proof}
It is sufficient to prove $\cS^{-1}\circ D_t\circ\cS=D_t$.
Note that $D_t$ is a derivation of the algebra of functions 
of the form $f=f(w^{pq}_{ij},\ud)$, and $\cS$ is an automorphism of this algebra.
Hence $\cS^{-1}\circ D_t\circ\cS$ is also a derivation of this algebra.

Since $D_t$ and $\cS^{-1}\circ D_t\circ\cS$ are derivations, 
to prove that $\cS^{-1}\circ D_t\circ\cS=D_t$, it is sufficient to show 
\begin{gather}
\lb{sdsul}
(\cS^{-1}\circ D_t\circ\cS)(u_l)=D_t(u_l)\qquad\quad\forall\,l,\\
\lb{sdsw}
(\cS^{-1}\circ D_t\circ\cS)(w^{pq}_{ij})=D_t(w^{pq}_{ij})
\qquad\quad\forall\,p,q,i,j.
\end{gather}
Equation~\eqref{sdsul} follows immediately from the definition of $\cS$ and $D_t$.
Equation~\eqref{sdsw} is equivalent to 
\begin{equation}
\lb{dssd}
(D_t\circ\cS)(w^{pq}_{ij})=(\cS\circ D_t)(w^{pq}_{ij})
\qquad\quad\forall\,p,q,i,j.
\end{equation}
Applying $D_t$ to~\eqref{csw} and $\cS$ to~\eqref{dtw}, and using~\eqref{lr},
we obtain~\eqref{dssd}.
\end{proof}

For any $h\in G$, we define the operator $\rs{h}$ on functions 
of the form $f(w^{pq}_{ij},\ud)$ as follows.
We set $\rs{h}(u_l)=u_l$ for all $l\in\zz$.
For a function~$y$ on an open subset $\mathbb{U}\subset G$, 
the function $\rs{h}(y)$ is defined on the open subset 
$\mathbb{U}h^{-1}\subset G$ by the formula
$$
\rs{h}(y)(g)=y(gh),\qquad g\in\mathbb{U}h^{-1}.
$$
Since $w^{pq}_{ij}$ are functions on $G$, we see that $\rs{h}(w^{pq}_{ij})$ 
is well defined.
Then for any function $f=f(w^{pq}_{ij},\ud)$ we set 
$\rs{h}(f)=f(\rs{h}(w^{pq}_{ij}),\ud)$.

Note that the operator $\rs{h}$ is invertible, and one has $\rs{h}^{-1}=\rs{h^{-1}}$.

\begin{lemma}
\lb{lemrs}
We have $\cS\circ\rs{h}=\rs{h}\circ\cS$ and 
$\rs{h}\circ D_t=D_t\circ\rs{h}$ for all $h\in G$.
\end{lemma}
\begin{proof}
To prove $\cS\circ\rs{h}=\rs{h}\circ\cS$, it is sufficient to show
that $(\cS\circ\rs{h})(w^{pq}_{ij})=(\rs{h}\circ\cS)(w^{pq}_{ij})$, 
which follows easily from~\eqref{csw} and the definition of~$\rs{h}$.
The main idea is the following. 
According to~\eqref{csw}, the action of $\cS$ is given 
by the left multiplication by $M(u,\la)$. 
The action of $\rs{h}$ is given by the right multiplication by $h$.
Since the left and right multiplications commute, we have $\cS\circ\rs{h}=\rs{h}\circ\cS$.

Let us prove $\rs{h}\circ D_t=D_t\circ\rs{h}$, which is equivalent to 
$\rs{h}^{-1}\circ D_t\circ\rs{h}=D_t$.
Note that $D_t$ is a derivation of the algebra of functions 
of the form $f=f(w^{pq}_{ij},\ud)$, and $\rs{h}$ is an automorphism of this algebra.
Therefore, $\rs{h}^{-1}\circ D_t\circ\rs{h}$ is also a derivation of this algebra.

Since $D_t$ and $\rs{h}^{-1}\circ D_t\circ\rs{h}$ are derivations, 
to prove that $\rs{h}^{-1}\circ D_t\circ\rs{h}=D_t$, it is sufficient to show 
\begin{gather}
\lb{rsdsul}
(\rs{h}^{-1}\circ D_t\circ\rs{h})(u_l)=D_t(u_l)\qquad\quad\forall\,l,\\
\lb{rsdsw}
(\rs{h}^{-1}\circ D_t\circ\rs{h})(w^{pq}_{ij})=D_t(w^{pq}_{ij})
\qquad\quad\forall\,p,q,i,j.
\end{gather}
Equation~\eqref{rsdsul} is obvious, because $\rs{h}(u_l)=u_l$ for all $l$.
Equation~\eqref{rsdsw} is equivalent to 
\begin{equation}
\notag
(D_t\circ\rs{h})(w^{pq}_{ij})=(\rs{h}\circ D_t)(w^{pq}_{ij})
\qquad\quad\forall\,p,q,i,j,
\end{equation}
which follows easily from~\eqref{dtw} and the definition of~$\rs{h}$.
\end{proof}

\begin{example}
\lb{exsdtw}
Consider the case $\sm=2$, $\og=1$, $k_1=2$ and 
the DLR~\eqref{lrvol} of the Volterra equation~\eqref{volt}.
Then $G=\gq{2}{\cp_1}{2}$, and formulas~\eqref{csw},~\eqref{dtw} become
\begin{multline}
\lb{vcsw}
\cS\begin{pmatrix}
 w^{10}_{11} & w^{10}_{12} \\
 w^{10}_{21} & w^{10}_{22} 
\end{pmatrix}
+(\lambda -\cp_1)
\cS\begin{pmatrix}
 w^{11}_{11} & w^{11}_{12} \\
 w^{11}_{21} & w^{11}_{22} 
\end{pmatrix}
=\\
=\begin{pmatrix}
 0 & u\\
 -1 & \lambda
 \end{pmatrix}
\left(
\begin{pmatrix}
 w^{10}_{11} & w^{10}_{12} \\
 w^{10}_{21} & w^{10}_{22} 
\end{pmatrix}
+(\lambda -\cp_1)
\begin{pmatrix}
 w^{11}_{11} & w^{11}_{12} \\
 w^{11}_{21} & w^{11}_{22} 
\end{pmatrix}
\right),
\qquad
(\la-\cp_1)^2=0,
\end{multline}
\begin{multline}
\lb{vdtw}
D_t\begin{pmatrix}
 w^{10}_{11} & w^{10}_{12} \\
 w^{10}_{21} & w^{10}_{22} 
\end{pmatrix}
+(\lambda -\cp_1)
D_t\begin{pmatrix}
 w^{11}_{11} & w^{11}_{12} \\
 w^{11}_{21} & w^{11}_{22} 
\end{pmatrix}
=\\
=\begin{pmatrix}
 u & \lambda u_{-1} \\
 -\lambda & \lambda^2+u_{-1} 
\end{pmatrix}
\left(
\begin{pmatrix}
 w^{10}_{11} & w^{10}_{12} \\
 w^{10}_{21} & w^{10}_{22} 
\end{pmatrix}
+(\lambda -\cp_1)
\begin{pmatrix}
 w^{11}_{11} & w^{11}_{12} \\
 w^{11}_{21} & w^{11}_{22} 
\end{pmatrix}\right),
\qquad
(\la-\cp_1)^2=0.
\end{multline}
\end{example}

Let $\mathcal{H}\subset G$ be a subgroup of $G$ and 
$\mathbb{U}\subset G$ be an open subset.
A function $f$ on $\mathbb{U}$ is called \emph{$\mathcal{H}$-right-invariant} 
if for any $g_1\in\mathcal{H}$ and $g_2\in\mathbb{U}$ such that 
$g_2g_1\in\mathbb{U}$ we have $f(g_2g_1)=f(g_2)$.

Similarly, $f$ is called \emph{$\mathcal{H}$-left-invariant} 
if for any $g_1\in\mathcal{H}$ and $g_2\in\mathbb{U}$ such that 
$g_1g_2\in\mathbb{U}$ we have $f(g_1g_2)=f(g_2)$.

Let $H\subset G$ be a closed connected Lie subgroup of $G$.
Consider the manifold $G/H$ and the natural projection
\begin{equation}
\notag
\pi\colon G\to G/H,\qquad\pi(g)=gH\in G/H,\qquad g\in G.
\end{equation}
Here we use the fact that the points of $G/H$ 
correspond to the left cosets $gH$ of $H$ in $G$.

\begin{remark}
\lb{remhri}
Let $Y\subset G/H$ be an open subset.
For each function $f$ on $Y$, one defines the function $\pi^*(f)$ 
on the open subset $\pi^{-1}(Y)\subset G$ as follows
$$
\pi^*(f)(g)=f(\pi(g))\qquad\quad\forall\,g\in\pi^{-1}(Y).
$$

A function $y$ on $\pi^{-1}(Y)$ is of the form $y=\pi^*(f)$ for some 
function $f$ on~$Y$ iff $y$ is $H$-right-invariant.
Therefore, functions on $Y\subset G/H$ can be regarded 
as $H$-right-invariant functions on $\pi^{-1}(Y)\subset G$.

Note that a function $y$ on $\pi^{-1}(Y)$ is $H$-right-invariant 
iff $\rs{h}(y)=y$ for all $h\in H$.

\end{remark}

Let $H\subset G$ be a closed connected Lie subgroup of codimension $r>0$.
Then $\dim G/H=r$.
Let $y^1,\dots,y^r$ be functions defined on 
an open dense subset of the manifold~$G/H$.
We say that $y^1,\dots,y^r$ 
\emph{form a system of local coordinates almost everywhere on~$G/H$}
if there is an open dense subset $\mathcal{Y}\subset G/H$ such that 
$y^1,\dots,y^r$ form a system of local coordinates on a neighborhood 
of each point of~$\mathcal{Y}$.

This is equivalent to the fact that 
the differentials of the functions $y^1,\dots,y^r$ 
are linearly independent almost everywhere on~$G/H$.

\begin{remark}
\lb{remlcgh}
In this remark, by a function on~$G/H$ we mean a function defined
on an open subset of~$G/H$, and similarly for functions on~$G$.
Suppose that $y^1,\dots,y^r$ form a system of local coordinates almost everywhere on~$G/H$.
Then any $H$-right-invariant function $\vf$ on~$G$ 
can locally be expressed as $\vf=\psi(y^1,\dots,y^r)$ for some function~$\psi$.
(Such an expression for $\vf$ exists locally almost everywhere, i.e., 
on neighborhoods of the points where 
$y^1,\dots,y^r$ form a system of local coordinates.)

Let $y$ be a function on~$G/H$.
According to Remark~\ref{remhri}, 
we can regard $y$ as an $H$-right-invariant function
on~$G$, so $y=y(w^{pq}_{ij})$ depends on the coordinates $w^{pq}_{ij}$ 
of~$G$. Then, according to the definition of the operators $\cS$ and $D_t$,
\begin{gather}
\lb{csyw}
\text{$\cS(y)$ may depend on $w^{pq}_{ij}$ and $u$},\\
\lb{dtyw}
\text{$D_t(y)$ may depend on $w^{pq}_{ij}$ and $u_\ta,\dots,u_\tb$}.
\end{gather}
Since $y$ is $H$-right-invariant and 
$\cS\circ\rs{h}=\rs{h}\circ\cS$, $\rs{h}\circ D_t=D_t\circ\rs{h}$ for all $h\in H$, 
the following property holds.

If we fix the values of the variables $u_l$ for all $l$ 
(including the variable $u=u_0$), 
then $\cS(y)$, $D_t(y)$ become $H$-right-invariant functions on~$G$ 
and can locally be expressed as functions of $y^1,\dots,y^r$.
Combining this observation with~\eqref{csyw} and~\eqref{dtyw}, 
we see that locally $\cS(y)$, $D_t(y)$ can be written in the form
\begin{equation}
\notag
\cS(y)=\mu(y^1,\dots,y^r,u),\quad\qquad 
D_t(y)=\xi(y^1,\dots,y^r,u_\ta,\dots,u_\tb)
\end{equation}
for some functions $\mu$, $\xi$.

\end{remark}

Recall that $\gla$ is the group of $\GL_\sm(\fik)$-valued functions of~$\la$, 
and the group $G$ is defined by~\eqref{ggr}.
Consider the natural homomorphism $\rho\colon\gla\to G$ given by
\begin{equation}
\lb{rhog}
\rho(\mathbf{g})=(g^1,\dots,g^{\og})\,\in\,\gq{\sm}{\cp_1}{k_1}\times\dots\times\gq{\sm}{\cp_\og}{k_\og}=G,
\qquad\quad\mathbf{g}\in\gla.
\end{equation}
Here $\mathbf{g}=\mathbf{g}(\la)$ is a 
$\GL_\sm(\fik)$-valued function of~$\la$, and for each $p=1,\dots,\og$ the element 
\begin{equation}
\lb{gprho}
g^p=\sum_{q=0}^{k_p-1}(\la-\cp_p)^qg^{p,q}\,\in\,\gq{\sm}{\cp_p}{k_p},\qquad
g^{p,0}\in\GL_\sm(\fik),\qquad
g^{p,1},\dots,g^{p,k_p-1}\in\gl_\sm(\fik),
\end{equation}
is determined by the Taylor expansion of $\mathbf{g}(\la)$ at the point $\la=\cp_p$.

Let $H\subset G$ be a closed connected Lie subgroup.
For any subgroup $\mathcal{H}\subset G$, we have defined the notion 
of $\mathcal{H}$-left-invariant functions on open subsets of~$G$.
Since $G$ acts by left multiplication on the manifold $G/H$, 
one can define similarly the notion of  
$\mathcal{H}$-left-invariant functions on open subsets of~$G/H$.

Recall that $\hg{k}\subset\gla$ for each $k\in\zp$.
Using the embedding $\rho(\hg{k})\subset G$, we can speak about 
$\hg{k}$-left-invariant functions on open subsets of $G$ and $G/H$.
So, by definition, a function is \emph{$\hg{k}$-left-invariant} 
if it is $\rho(\hg{k})$-left-invariant.

\begin{lemma}
\lb{lemyh}
Let $y=y(w^{pq}_{ij})$ be a function on an open dense subset of~$G$.
According to~\eqref{csw}, the function $\cS(y)=y(\cS(w^{pq}_{ij}))$ may depend 
on $w^{pq}_{ij}$, $u$.
If $y$ is $\hg{k}$-left-invariant for some $k>0$, then
$\cS(y)$ does not depend on $u$ and is $\hg{k-1}$-left-invariant. 
\end{lemma}
\begin{proof}

In what follows, we use the symbols $g$ and $\hat g$ as arguments of some functions 
defined on open dense subsets of $G$.
When we write $g\in G$ (or $\hat g\in G$), we mean that $g$ (or $\hat g$) 
belongs to an open subset where the considered function is defined.

For example, since $y$ is defined on an open dense subset of~$G$, 
for any fixed values of ${\tilde u}$, $u$ there is an open dense 
subset $\mathbb{U}\subset G$ such that 
$y(\hat g)$ and $y\big(\rho(M({\tilde u},\la)\cdot M(u,\la)^{-1})\hat g\big)$ 
are defined for all $\hat g\in\mathbb{U}$. In equation~\eqref{zgzg} below 
we assume that $\hat g$ belongs to such an open dense subset, 
so that the left-hand side and right-hand side of~\eqref{zgzg} are well defined.
Similar considerations apply also to other equations in this proof.

Since $\cS(y)=y(\cS(w^{pq}_{ij}))$ may depend on $w^{pq}_{ij}$ and $u$, 
for each fixed value of $u$ we can regard $\cS(y)$ as a function 
on an open dense subset of~$G$. 
Formula~\eqref{csw} implies that this interpretation of $\cS(y)$ can be written as
\begin{equation}
\lb{cszg}
\cS(y)(g)=y\big(\rho(M(u,\la))g\big),\qquad\quad g\in G.
\end{equation}

Recall that $\hg{1}$ is generated by the elements~\eqref{mumu}, 
and we have $\hg{1}\subset\hg{k}$ for $k>0$.
Therefore, the elements~\eqref{mumu} belong to $\hg{k}$. 
Since $y$ is defined on an open dense subset of~$G$
and is $\hg{k}$-left-invariant, for any fixed values 
of ${\tilde u}$, $u$ we have 
\begin{equation}
\lb{zgzg}
y\big(\rho(M({\tilde u},\la)\cdot M(u,\la)^{-1})\hat g\big)=y(\hat g),
\qquad\quad\hat g\in G.
\end{equation}
Taking $g=\rho(M(u,\la)^{-1})\hat g$, from~\eqref{zgzg} we obtain
\begin{equation}
\lb{zmzm}
y\big(\rho(M({\tilde u},\la))g\big)=
y\big(\rho(M(u,\la))g\big),\qquad\quad g\in G.
\end{equation}
Combining~\eqref{zmzm} with~\eqref{cszg}, we see that $\cS(y)$
does not depend on $u$.

Let us show that $\cS(y)$ is $\hg{k-1}$-left-invariant. 
Recall that for any $h\in\hg{k-1}$ and any fixed value of~$u$ one has
$M(u,\la)\cdot h\cdot M(u,\la)^{-1}\in\hg{k}$.
Since $y$ is defined on an open dense subset of~$G$
and is $\hg{k}$-left-invariant, we have 
\begin{equation}
\lb{zmhm}
y\big(\rho(M(u,\la)\cdot h\cdot M(u,\la)^{-1})\hat g\big)=y(\hat g),
\qquad\quad\hat g\in G.
\end{equation}
Taking $g=\rho(M(u,\la)^{-1})\hat g$, from~\eqref{zmhm} we get
\begin{equation}
\lb{yy}
y\big(\rho(M(u,\la)\cdot h)g\big)=
y\big(\rho(M(u,\la))g\big),\qquad\quad g\in G.
\end{equation}
Combining~\eqref{yy} with~\eqref{cszg}, one obtains
$$
\cS(y)\big(\rho(h)g\big)=\cS(y)(g),\qquad g\in G,\qquad h\in\hg{k-1},
$$
which says that $\cS(y)$ is $\hg{k-1}$-left-invariant. 
\end{proof}

\begin{lemma}
\lb{suiz}
Let $H\subset G$ be a closed connected Lie subgroup of~$G$.
Let $z$ be a function on an open dense subset of~$G/H$.
According to Remark~\ref{remhri}, the function~$z$ can also be regarded 
as an $H$-right-invariant function on an open dense subset of~$G$.
Therefore, we can consider the functions $\cS^k(z)$, $k\in\zsp$,
which may depend on~$w^{pq}_{ij}$ and $u_l$.

Suppose that $z$ is $\hg{\al}$-left-invariant for some $\al\in\zsp$.
Then the functions 
\begin{equation}
\lb{stk}
\cS^k(z),\qquad\quad k=1,\dots,\al,
\end{equation}
do not depend on $u_l$ for any $l\in\zz$.
Furthermore, the functions~\eqref{stk} are $H$-right-invariant 
and can be viewed as functions on an open dense subset of~$G/H$.
\end{lemma}
\begin{proof}
Applying Lemma~\ref{lemyh} to $y=z$, we see that $\cS(z)$ 
does not depend on $u_l$ for any $l\in\zz$ and is $\hg{\al-1}$-left-invariant.
If $\al\ge 2$, 
then, applying Lemma~\ref{lemyh} to $y=\cS(z)$, one obtains that the function 
$\cS^2(z)=\cS(\cS(z))$ 
does not depend on $u_l$ for any $l\in\zz$ and is $\hg{\al-2}$-left-invariant.

Similarly, using Lemma~\ref{lemyh}, 
by induction on $k=1,\dots,\al$ we prove that the function 
$\cS^k(z)$ does not depend on $u_l$ for any $l\in\zz$ and 
is $\hg{\al-k}$-left-invariant.

Since $z$ is $H$-right-invariant and is defined on an open dense subset of~$G/H$, 
we have also the same property for the functions~\eqref{stk}, 
because $\cS\circ\rs{h}=\rs{h}\circ\cS$ for all $h\in H$.
\end{proof}

\subsection{The method to construct MTs}
\lb{smcmt}

Recall that the Lie group $G$ is given by~\eqref{ggr}.
We are going to describe how to construct MTs from some functions on open subsets of $G/H$, 
where $H\subset G$ is a closed connected Lie subgroup.

Recall that $u_l=(u^1_l,\dots,u^{\diu}_l)$ is an $\diu$-component vector for each $l\in\zz$.
To clarify the main idea, we consider first the case $\diu=1$.
(The case of arbitrary $\diu\in\zsp$ will be described in Theorem~\ref{tharb}.)

Note that Theorems~\ref{thscal},~\ref{tharb} do not mention the DLR~\eqref{mu} 
explicitly, but the matrices $M$, $\lt$ from~\eqref{mu} appear in the definition~\eqref{csw}, \eqref{dtw} 
of the operators $\cS$, $D_t$, which are used in Theorems~\ref{thscal},~\ref{tharb}.

In the case $\diu=1$ the vector $u_l$ 
has only one component $u^1_l$, which we denote by the same symbol $u_l$. 
According to our notation, $u=u_0$.

\begin{theorem}
\lb{thscal}
Suppose that $\diu=1$. 
Recall that the Lie group $G$ is defined by~\eqref{ggr}.
Let $H\subset G$ be a closed connected Lie subgroup of codimension $r>0$.

Let $z$ be an $\hg{r-1}$-left-invariant function 
on an open dense subset of the manifold~$G/H$. 
\textup{(}According to Remark~\ref{remhri}, the function~$z$ can also be regarded 
as an $H$-right-invariant 
and $\hg{r-1}$-left-invariant function on an open dense subset of~$G$.\textup{)}

By Lemma~\ref{suiz}, the functions $\cS^k(z)$, $k=1,\dots,r-1$, 
do not depend on $u_l$, and we can view $\cS^k(z)$, $k=1,\dots,r-1$, 
as functions on an open dense subset of~$G/H$. 
Suppose that
\begin{gather}
\lb{sclc}
\begin{array}{c}
\text{the functions $z,\ \cS(z),\ \cS^2(z),\,\dots,\ \cS^{r-1}(z)$ form}\\
\text{a system of local coordinates almost everywhere on~$G/H$}\\
\text{\textup{(}i.e., the differentials of the functions 
are linearly independent almost everywhere\textup{)}},
\end{array}\\
\lb{pdu}
\text{the function $\frac{\partial}{\partial u}\big(\cS^{r}(z)\big)$
is not identically zero}.
\end{gather}

Set $z_0=z$ and $z_k=\cS^k(z)$ for $k=1,\dots,r$.
Condition~\eqref{sclc} says that $z_0,z_1,\dots,z_{r-1}$ form a system of local coordinates
almost everywhere on $G/H$.
Applying Remark~\ref{remlcgh} to this system of local coordinates, 
we see that locally one has 
\begin{gather}
\lb{zrfzu}
z_r=\cS^{r}(z)=\cS(z_{r-1})=\fe(z_0,z_1,\dots,z_{r-1},u),\\
\notag
D_t(z)=Q(z_0,\dots,z_{r-1},u_\ta,\dots,u_\tb)
\end{gather}
for some functions $\fe$ and $Q$.
By the implicit function theorem, condition~\eqref{pdu} implies that 
locally from equation~\eqref{zrfzu} we can express $u$ in terms of $z_0,z_1,\dots,z_r$
\begin{equation}
\lb{ufz}
u=\mf(z_0,z_1,\dots,z_{r}).
\end{equation}
\textup{(}We can do this on a neighborhood of a point where 
$\dfrac{\partial}{\partial u}\big(\cS^{r}(z)\big)\neq 0$.\textup{)}

Then one obtains an MT for equation~\eqref{sdde} as follows.
We introduce new variables $\fv_l$ for $l\in\mathbb{Z}$.
Using the function $\mf(z_0,z_1,\dots,z_{r})$ from~\eqref{ufz}, 
for each $j\in\mathbb{Z}$ we can consider the function 
$\mf(\fv_{j},\fv_{j+1},\dots,\fv_{j+r})$.
Let $P(\fv_{\ta},\dots,\fv_{{\tb}+r})$ be the function 
obtained from $Q(z_0,\dots,z_{r-1},u_\ta,\dots,u_\tb)$
by replacing $z_i$ with $\fv_i$ and $u_j$ with $\mf(\fv_{j},\fv_{j+1},\dots,\fv_{j+r})$.
That is,
\begin{equation}
\lb{pq}
P(\fv_{\ta},\dots,\fv_{{\tb}+r})
=Q\big(\fv_0,\dots,\fv_{r-1},\mf(\fv_{\ta},\dots,\fv_{\ta+r}),\dots,
\mf(\fv_{\tb},\dots,\fv_{\tb+r})\big).
\end{equation}

We introduce the formula
\begin{equation}
\lb{vtp}
v_t=P(\fv_{\ta},\dots,\fv_{{\tb}+r})
\end{equation}
and regard~\eqref{vtp} as a differential-difference equation for the variable~$\fv$.
Then the formula 
\begin{equation}
\lb{umffv}
u=\mf(\fv_0,\fv_1,\dots,\fv_{r})
\end{equation}
determines an MT from equation~\eqref{vtp} to equation~\eqref{sdde}. 

As usual, we can use the identification $v_0=v$. 
Then \eqref{umffv} becomes $u=\mf(\fv,\fv_1,\dots,\fv_{r})$.
\end{theorem}
\begin{proof}
As has been said above, we can regard $z$ as 
a function on an open dense subset of~$G$, so $z=z(w^{pq}_{ij})$ 
is a function of the coordinates $w^{pq}_{ij}$ on~$G$.
Set $z_k=\cS^k(z)$ for all $k\in\mathbb{Z}$.
The functions $z_k$ may depend on $w^{pq}_{ij}$ and $u_l$.

\begin{lemma}
\lb{lemfi}
The functions $z_k$, $k\in\mathbb{Z}$, are functionally independent.
\end{lemma}
\begin{proof}
Suppose that $z_k$, $k\in\mathbb{Z}$, are not functionally independent.
This means that there are $p,q\in\zz$, $p\le q$, and a nontrivial relation of the form
\begin{equation}
\lb{rz}
R(z_{p},z_{p+1},\dots,z_{q})=0.
\end{equation}
Applying $\cS^{-p}$ to equation~\eqref{rz} and using the identities $\cS^j(z_k)=z_{k+j}$ for $j\in\zz$,
we obtain 
$$
R(z_{0},z_{1},\dots,z_{q-p})=0, 
$$
which implies that 
\begin{equation}
\lb{zzpq}
z_{0},\quad z_{1},\quad \dots,\quad z_{q-p}
\end{equation}
are not functionally independent.
The functions $z_{0},\dots,z_{r-1}$ form a system of local coordinates 
almost everywhere on~$G/H$, 
hence $z_{0},\dots,z_{r-1}$ are functionally independent.
Therefore, ${q-p\ge r}$.

For any $j\in\zz$, applying $\cS^j$ to equations~\eqref{zrfzu},~\eqref{ufz} 
and using the identities $\cS^j(z_k)=z_{k+j}$, $\cS^j(u)=u_j$,
one gets
\begin{gather}
\lb{zrjf}
z_{r+j}=\fe(z_j,z_{j+1},\dots,z_{j+r-1},u_j),\\
\lb{ujfz}
u_j=\mf(z_j,z_{j+1},\dots,z_{j+r})\qquad\quad\forall\,j\in\zz.
\end{gather}
Equations~\eqref{zrjf},~\eqref{ujfz} imply that the functions~\eqref{zzpq}
can be expressed in terms of
\begin{equation}
\lb{zupq}
z_{0},\quad z_{1},\quad \dots,\quad z_{r-1},\quad u,\quad u_{1},\quad\dots,\quad u_{q-p-r},
\end{equation}
and the functions~\eqref{zupq} can be expressed in terms of~\eqref{zzpq}.
Since~\eqref{zupq} are functionally independent, we obtain that 
\eqref{zzpq} are functionally independent as well.
\end{proof}

For any $j,\,i_1,\dots,i_l\in\zz$ and any function 
of the form $h=h(\fv_{i_1},\dots,\fv_{i_l})$, 
we set 
$$
\cS^j(h)=h(\fv_{i_1+j},\dots,\fv_{i_l+j}).
$$
According to Definition~\ref{defmtt} of MTs, 
to prove the theorem, we need to show that
\begin{multline}
\lb{mteq}
\sum_{i=0}^r\frac{\pd}{\pd\fv_i}
\big(\mf(\fv_0,\fv_1,\dots,\fv_{r})\big)\cdot
\cS^i\big(P(\fv_{\ta},\dots,\fv_{{\tb}+r})\big)=\\
=\ff\big(\mf(\fv_{\ua},\dots,\fv_{\ua+r}),
\mf(\fv_{\ua+1},\dots,\fv_{\ua+1+r}),\dots,
\mf(\fv_{\ub},\dots,\fv_{\ub+r})\big).
\end{multline}
Applying $D_t$ to equation~\eqref{ufz} 
and using the identity~$D_t(u)=\ff(u_\ua,u_{\ua+1},\dots,u_\ub)$, one obtains
\begin{equation}
\lb{ffumf}
\ff(u_\ua,u_{\ua+1},\dots,u_\ub)=
\sum_{i=0}^r\frac{\pd}{\pd z_i}
\big(\mf(z_0,z_1,\dots,z_{r})\big)\cdot
D_t(z_i).
\end{equation}
Since $\cS\circ D_t=D_t\circ\cS$, we have
\begin{equation}
\lb{dtzi}
D_t(z_i)=D_t(\cS^i(z))=\cS^i(D_t(z))=
\cS^i(Q(z_0,\dots,z_{r-1},u_\ta,\dots,u_\tb)).
\end{equation}
Substituting~\eqref{ujfz} and~\eqref{dtzi} in~\eqref{ffumf}, we obtain
\begin{multline}
\lb{fffmf}
\ff\big(\mf(z_{\ua},\dots,z_{\ua+r}),
\mf(z_{\ua+1},\dots,z_{\ua+1+r}),\dots,
\mf(z_{\ub},\dots,z_{\ub+r})\big)=\\
=\sum_{i=0}^r\frac{\pd}{\pd z_i}
\big(\mf(z_0,z_1,\dots,z_{r})\big)\cdot
\cS^i\big(Q\big(z_0,\dots,z_{r-1},\mf(z_{\ta},\dots,z_{\ta+r}),\dots,
\mf(z_{\tb},\dots,z_{\tb+r})\big)\big).
\end{multline}
Since $z_j$, $j\in\mathbb{Z}$, are functionally independent, 
the identity~\eqref{fffmf} will remain valid if we replace $z_j$ by $\fv_j$ for all $j$.
Replacing $z_j$ by $\fv_j$ in~\eqref{fffmf} and using~\eqref{pq}, we get~\eqref{mteq}.
\end{proof}

\begin{remark}
\lb{nondeg}
Informally speaking, \eqref{sclc} and \eqref{pdu} can be regarded as non-degeneracy conditions.
Condition~\eqref{sclc} says that 
the differentials of the functions $z$, $\cS(z)$, $\cS^2(z),\dots,\cS^{r-1}(z)$
are linearly independent almost everywhere on the $r$-dimensional manifold~$G/H$.
Condition~\eqref{pdu} says that the function 
$\cS^{r}(z)=\fe(z_0,z_1,\dots,z_{r-1},u)$ depends nontrivially on $u$.

In constructing an MT by the method described in Theorem~\ref{thscal},
the most important step is to find a nonconstant 
$\hg{r-1}$-left-invariant function $z$ on an open dense subset of~$G/H$.
For such a function $z$, one can expect that conditions~\eqref{sclc},~\eqref{pdu} are usually satisfied.
This is the case in all examples known to us.

\end{remark}

\begin{example}
\lb{m1vol}
Consider the case $\sm=2$, $\og=1$, $k_1=2$ and 
the DLR~\eqref{lrvol} of the Volterra equation~\eqref{volt}.
The group $G=\gq{2}{\cp_1}{2}$ can be described as follows
\begin{equation}
\lb{g22}
G=\gq{2}{\cp_1}{2}=\left.\left\{
\begin{pmatrix}
 w^{10}_{11} & w^{10}_{12} \\
 w^{10}_{21} & w^{10}_{22} 
\end{pmatrix}
+(\lambda -\cp_1)
\begin{pmatrix}
 w^{11}_{11} & w^{11}_{12} \\
 w^{11}_{21} & w^{11}_{22} 
\end{pmatrix}\ \right|\ 
w^{10}_{11}w^{10}_{22}-w^{10}_{21}w^{10}_{12}\neq 0\right\},
\end{equation}
where $\cp_1\in\fik$ is a fixed constant, 
and $w^{10}_{11}$, $w^{10}_{12}$, $w^{10}_{21}$, $w^{10}_{22}$, 
$w^{11}_{11}$, $w^{11}_{12}$, $w^{11}_{21}$, $w^{11}_{22}$ are coordinates on~$G$. 
Since $k_1=2$, we work modulo the relation ${(\la-\cp_1)^2=0}$.

Consider the subgroup
\begin{equation}
\notag
H=\left.\left\{
\begin{pmatrix}
 1 & w^{10}_{12} \\
 0 & w^{10}_{22} 
\end{pmatrix}
+(\lambda -\cp_1)
\begin{pmatrix}
 w^{11}_{11} & w^{11}_{12} \\
 w^{11}_{21} & w^{11}_{22} 
\end{pmatrix}\ \right|\ 
w^{10}_{22}\neq 0\right\}\,\subset\,G,
\end{equation}
which is of codimension $r=2$.
As has been shown in Example~\ref{exlrv}, for the DLR~\eqref{lrvol},
the group $\hg{1}$ consists of the constant matrix-functions 
$\begin{pmatrix}
 a_1 & 0 \\
 0 & 1 
\end{pmatrix}$.

For any $h\in H$ and $h_1\in\hg{1}$, in the product 
$$
h_1\cdot
\left(
\begin{pmatrix}
 w^{10}_{11} & w^{10}_{12} \\
 w^{10}_{21} & w^{10}_{22} 
\end{pmatrix}
+(\lambda -\cp_1)
\begin{pmatrix}
 w^{11}_{11} & w^{11}_{12} \\
 w^{11}_{21} & w^{11}_{22} 
\end{pmatrix}
\right)
\cdot h
$$
the $w^{10}_{21}$ component does not change.
Therefore, the function $z=w^{10}_{21}$ is 
$H$-right-invariant and $\hg{1}$-left-invariant.

Using formula~\eqref{vcsw} and the notation $z_k=\cS^k(z)$ for $k\in\zz$, we obtain
\begin{gather}
\lb{z0z1}
z_0=z=w^{10}_{21},\qquad
z_1=\cS(z)=\cS(w^{10}_{21})=\cp_1 w_{21}^{10}-w_{11}^{10},\\
\lb{z2cs2}
z_2=\cS^2(z)=\cS(z_1)=\cp_1\cS(w_{21}^{10})-\cS(w_{11}^{10})
=(\cp_1^2-u)w^{10}_{21}-\cp_1 w^{10}_{11}.
\end{gather}
Recall that in this example we have $r=2$.
Since the differentials of the functions $z$ and $\cS(z)$ are linearly independent, 
condition~\eqref{sclc} is valid.
Equation~\eqref{z2cs2} shows that condition~\eqref{pdu} is valid as well.

Therefore, using Theorem~\ref{thscal}, we obtain an MT as follows.
Equations~\eqref{z0z1}, \eqref{z2cs2} imply $z_2=\cp_1z_1-uz_0$, which yields
\begin{equation}
\lb{u0zz}
u=\frac{\cp_1z_1-z_2}{z_0}.
\end{equation}
Using~\eqref{vdtw},~\eqref{z0z1},~\eqref{z2cs2}, we get
\begin{equation}
\lb{dtzu}
D_t(z)=D_t(w^{10}_{21})=\cp_1^2w_{21}^{10}-\cp_1 w_{11}^{10}+u_{-1} w_{21}^{10}
=\cp_1z_1+u_{-1}z_0.
\end{equation}
Equation~\eqref{u0zz} implies 
$u_{-1}=\cS^{-1}(({\cp_1z_1-z_2})/{z_0})=({\cp_1z_0-z_1})/{z_{-1}}$.
Substituting this in~\eqref{dtzu}, one obtains
\begin{equation}
\lb{dtzz}
D_t(z)=\frac{\cp_1 z_0^2+\cp_1 z_{-1} z_1-z_1 z_0}{z_{-1}}.
\end{equation}
According to Theorem~\ref{thscal}, to obtain an MT, 
we need to replace $z_k$ by $v_k$ for all $k\in\zz$ in~\eqref{u0zz},~\eqref{dtzz}.
(And we can use the identification $v_0=v$.) 

Thus we get the following result. For any $\cp_1\in\fik$, 
the formula 
\begin{equation}
\lb{ucv1v2}
u=\frac{\cp_1v_1-v_2}{v}
\end{equation}
determines an MT from the equation 
\begin{equation}
\lb{vtcpv1}
v_t=\frac{\cp_1 v^2+\cp_1 v_{-1} v_1-v_1 v}{v_{-1}}
\end{equation}
to the Volterra equation~\eqref{volt}.
As we show in Remark~\ref{mtyam1}, this MT can be found 
in~\cite{yam2006}, after some change of variables.

Note that the same MT can be obtained also in the case
$\sm=2$, $\og=1$, $k_1=1$. In this case, the group~\eqref{ggr} is equal to 
\begin{equation}
\notag
\gq{2}{\cp_1}{1}=\left.\left\{
\begin{pmatrix}
 w^{10}_{11} & w^{10}_{12} \\
 w^{10}_{21} & w^{10}_{22} 
\end{pmatrix}\ \right|\ 
w^{10}_{11}w^{10}_{22}-w^{10}_{21}w^{10}_{12}\neq 0\right\}.
\end{equation}
Consider the subgroup
$\tilde H=\left.\left\{
\begin{pmatrix}
 1 & w^{10}_{12} \\
 0 & w^{10}_{22} 
\end{pmatrix}\ \right|\ 
w^{10}_{22}\neq 0\right\}\,\subset\,\gq{2}{\cp_1}{1}$.
The function $z=w^{10}_{21}$ is 
$\tilde H$-right-invariant and $\hg{1}$-left-invariant, 
and it gives the same MT.
\end{example}
\begin{remark}
\lb{mtyam1}
After the change of variables $v=\exp{\tilde v},\ t=-\tilde t$,
equation~\eqref{vtcpv1} becomes
\begin{equation}
\label{wv1}
\pd_{\tilde t}\big({\tilde v}\big)
=\big(\exp({\tilde v}_1-{\tilde v})-\cp_1\big)
\big(\exp({\tilde v}-{\tilde v}_{-1})-\cp_1\big)-\cp_1^2.
\end{equation}
Equation~\eqref{wv1} is a particular case of equation~(V6)
from the list of Volterra-type equations in~\cite{yam2006}.
The paper~\cite{yam2006} presents also an MT reducing this equation 
to the Volterra equation. The MT~\eqref{ucv1v2} coincides with 
the MT from~\cite{yam2006}, up to the above change of variables.
\end{remark}

\begin{example}
\lb{m2vol}
Let $\sm=2$, $\og=1$, $k_1=2$.
Consider the group $G=\gq{2}{\cp_1}{2}$ given by~\eqref{g22},
the subgroup 
\begin{equation}
H=\left.\left\{
\begin{pmatrix}
 w^{10}_{11} & w^{10}_{12} \\
 0 & w^{10}_{22} 
\end{pmatrix}
+(\lambda -\cp_1)
\begin{pmatrix}
 w^{11}_{11} & w^{11}_{12} \\
 0 & w^{11}_{22} 
\end{pmatrix}\ \right|\ 
w^{10}_{11}\neq 0,\ 
w^{10}_{22}\neq 0\right\}\,\subset\,G,
\end{equation}
and the DLR~\eqref{lrvol} of the Volterra equation~\eqref{volt}.
Clearly, the subgroup $H\subset G$ is of codimension $r=2$.

It is straightforward to check that the functions
\begin{equation}
\lb{hfs}
y^1=\frac{w^{10}_{11}}{w^{10}_{21}},\qquad\quad
y^2=\frac{w^{10}_{21}w^{11}_{11}-w^{10}_{11}w^{11}_{21}}{(w^{10}_{21})^2}
\end{equation}
are $H$-right-invariant.
Therefore, any function of the form $\al(y^1,y^2)$ is also $H$-right-invariant. 

According to Example~\ref{exlrv}, for the DLR~\eqref{lrvol},
the group $\hg{1}$ consists of the constant matrix-functions 
$\begin{pmatrix}
 a_1 & 0 \\
 0 & 1 
\end{pmatrix}$.
To apply Theorem~\ref{thscal}, we need to find a function $z=\al(y^1,y^2)$ 
such that $z$ is $\hg{1}$-left-invariant.

The group $\hg{1}$ acts on $G$ by left multiplication as follows
\begin{multline}
\notag
\begin{pmatrix}
 a_1 & 0 \\
 0 & 1 
\end{pmatrix}\cdot
\left(
\begin{pmatrix}
 w^{10}_{11} & w^{10}_{12} \\
 w^{10}_{21} & w^{10}_{22} 
\end{pmatrix}
+(\lambda -\cp_1)
\begin{pmatrix}
 w^{11}_{11} & w^{11}_{12} \\
 w^{11}_{21} & w^{11}_{22} 
\end{pmatrix}\right)=\\
=\begin{pmatrix}
 a_1 w^{10}_{11} & a_1 w^{10}_{12} \\
 w^{10}_{21} & w^{10}_{22} 
\end{pmatrix}
+(\lambda -\cp_1)
\begin{pmatrix}
 a_1 w^{11}_{11} & a_1 w^{11}_{12} \\
 w^{11}_{21} & w^{11}_{22} 
\end{pmatrix},
\end{multline}
which implies that $z=y^1/y^2$ is $\hg{1}$-left-invariant.
From~\eqref{vcsw},~\eqref{hfs} it follows that 
\begin{gather}
\lb{sy1}
\cS(y^1)=\cS\left(\frac{w^{10}_{11}}{w_{21}^{10}}\right)
=\frac{\cS(w_{11}^{10})}{\cS(w_{21}^{10})}
=\frac{uw_{21}^{10}}{\cp_1 w^{10}_{21}-w_{11}^{10}}
=\frac{u}{\cp_1-y^1},\\
\lb{sy2}
\cS(y^2) = 
\frac{\cS(w_{21}^{10})\cS(w_{11}^{11})-\cS(w_{11}^{10})\cS(w_{21}^{11})}{\cS(w_{21}^{10})^2}
=
\frac{u(w_{21}^{10}w_{11}^{11}-w_{11}^{10}w_{21}^{11}-w_{21}^{10}w_{21}^{10})}{(\cp_1 w_{21}^{10}-w_{11}^{10})^2} 
=\frac{u (y^2-1)}{(\cp_1-y^1)^2}.
\end{gather}
Using~\eqref{sy1},~\eqref{sy2} and the notation $z_k=\cS^k(z)$ for $k\in\zz$, we obtain
\begin{gather}
\lb{yz0z1}
z_0=z=\frac{y^1}{y^2},\qquad\quad
z_1=\cS(z)=\frac{\cS(y^1)}{\cS(y^2)}=\frac{\cp_1-y^1}{y^2-1},\\
\lb{yz2cs2}
z_2=\cS^2(z)=\cS(z_1)=
\frac{\cp_1-\cS(y^1)}{\cS(y^2)-1}=
\frac{(y^1-\cp_1)(\cp_1 (\cp_1-y^1)-u)}{(\cp_1-y^1)^2-(y^2-1)u}.
\end{gather}
Recall that in this example we have $r=2$.
Since the differentials of the functions $z$ and $\cS(z)$ 
are linearly independent, 
condition~\eqref{sclc} is valid.
Equation~\eqref{yz2cs2} shows that condition~\eqref{pdu} is valid as well.

From~\eqref{yz0z1},~\eqref{yz2cs2} one gets 
\begin{gather}
\lb{y12z}
y^1=\frac{z_0(\cp_1+z_1)}{z_0+z_1},\quad\qquad 
y^2=\frac{\cp_1+z_1}{z_0+z_1},\\
\lb{u0yz}
u=\frac{z_{1}^2 (\cp_1-z_0)(\cp_1+z_{2})}{(z_0+z_{1})(z_{1}+z_{2})}.
\end{gather}
Since $z=y^1/y^2$, we have 
\begin{equation}
\lb{dtzy}
D_t(z)=\frac{D_t(y^1)y^2-D_t(y^2)y^1}{(y^2)^2}.
\end{equation}
To compute $D_t(y^1)$ and $D_t(y^2)$ in~\eqref{dtzy}, one can use formulas~\eqref{hfs},~\eqref{vdtw}. 
This gives 
\begin{equation}
\lb{dtzyu}
D_t(z)=\frac{u_{-1}(\cp_1 y^2-y^1)-(y^1)^2(\cp_1(y^2-2)+y^1)}{(y^2)^2}.
\end{equation}
Equation~\eqref{u0yz} implies 
$u_{-1}=\cS^{-1}\left(\dfrac{z_{1}^2 (\cp_1-z_0)(\cp_1+z_{2})}{(z_0+z_{1})(z_{1}+z_{2})}\right)=
\dfrac{z_{0}^2 (\cp_1-z_{-1})(\cp_1+z_{1})}{(z_{-1}+z_{0})(z_{0}+z_{1})}$.
Substituting this and~\eqref{y12z} in~\eqref{dtzyu}, we obtain 
\begin{equation}
\lb{dtzzz}
D_t(z)=
\frac{z_0^2(z_{-1}-z_1)(z_0-\cp_1)(z_0+\cp_1)}{(z_0+z_{-1})(z_0+z_1)}.
\end{equation}

According to Theorem~\ref{thscal}, to obtain an MT, 
we need to replace $z_k$ by $v_k$ for all $k\in\zz$ in~\eqref{u0yz},~\eqref{dtzzz}.
Thus we get the following result. For any $\cp_1\in\fik$, 
the formula 
\begin{equation}
\lb{uv12v}
u=\frac{v_{1}^2 (\cp_1-v)(\cp_1+v_{2})}{(v+v_{1})(v_{1}+v_{2})}
\end{equation}
determines an MT from the equation 
\begin{equation}
\lb{vtvvv}
v_t=\frac{v^2(v_{-1}-v_1)(v-\cp_1)(v+\cp_1)}{(v+v_{-1})(v+v_1)}
\end{equation}
to the Volterra equation~\eqref{volt}.

Equation~\eqref{vtvvv} is a particular case of equation~(V2)
from the list of Volterra-type equations in~\cite{yam2006}.
On page~599 of~\cite{yam2006} Yamilov presents an MT connecting equation~(V2)
with the Volterra equation. Formula~\eqref{uv12v} is a particular case 
of the MT from~\cite{yam2006} (up to a shift of~$u$).
\end{example}

Now consider the case of arbitrary $\diu\in\zsp$.
Recall that $u_l=(u^1_l,\dots,u^{\diu}_l)$ for each $l\in\zz$ 
and $u^j=u^j_0$ for each $j=1,\dots,\diu$. 

\begin{theorem}
\lb{tharb}
Recall that the Lie group $G$ is defined by~\eqref{ggr}.
Let $H\subset G$ be a closed connected Lie subgroup of codimension $r>0$.

Suppose that there are nonnegative integers $\da_1,\dots,\da_\diu$ 
and functions $z^1,\dots,z^\diu$ on an open dense subset of $G/H$ such that 
$\da_1+\dots+\da_\diu+\diu=r$, the function $z^i$ is 
$\hg{\da_i}$-left-invariant for each $i=1,\dots,\diu$, and the following conditions hold
\begin{gather}
\lb{msclc}
\begin{array}{c}
\text{the functions $\cS^k(z^i)$, $i=1,\dots,\diu$, $k=0,\dots,\da_i$, form}\\
\text{a system of local coordinates almost everywhere on $G/H$}\\
\text{\textup{(}i.e., the differentials of the functions 
are linearly independent almost everywhere\textup{)}},
\end{array}\\
\lb{mpdu}
\begin{array}{c}
\text{the determinant of the $\diu\times\diu$ matrix-function 
$\left(\dfrac{\partial}{\partial u^j}\big(\cS^{\da_i+1}(z^i)\big)\right)$}\\
\text{is not identically zero.}
\end{array}
\end{gather}
\textup{(}Note that, by Lemma~\ref{suiz}, since $z^i$ is $\hg{\da_i}$-left-invariant, 
the functions $\cS^k(z^i)$, $i=1,\dots,\diu$, $k=0,\dots,\da_i$, 
do not depend on $u_l$ and are well defined on an open dense subset 
of $G/H$, so condition~\eqref{msclc} makes sense.\textup{)}

Set $z^i_0=z^i$ and $z^i_k=\cS^k(z^i)$ for $i=1,\dots,\diu$ and $k=1,\dots,\da_i+1$.
Condition~\eqref{msclc} says that $z^i_k$, $i=1,\dots,\diu$, $k=0,\dots,\da_i$, 
form a system of local coordinates almost everywhere on~$G/H$.
Applying Remark~\ref{remlcgh} to this system of local coordinates, 
we see that locally one has 
\begin{gather}
\lb{mzrfzu}
z^i_{\da_i+1}=\cS^{\da_i+1}(z^i)=\cS(z^i_{\da_i})=\fe^i(z^{i'}_{k'},u^{i'}),\qquad
i,i'=1,\dots,\diu,\qquad 
k'=0,\dots,\da_{i'},\\
\notag
D_t(z^i)=Q^i(z^{i'}_{k'},u_\ta^{i'},\dots,u_\tb^{i'}),\qquad
i,i'=1,\dots,\diu,\qquad 
k'=0,\dots,\da_{i'},
\end{gather}
for some functions $\fe^i$ and $Q^i$. 
By the implicit function theorem, condition~\eqref{mpdu} implies that 
locally from equations~\eqref{mzrfzu} we can express $u^1,\dots,u^\diu$ 
in terms of $z^{\ti}_{\tk}$, $\ti=1,\dots,\diu$, $\tk=0,\dots,\da_{\ti}+1$, 
\begin{equation}
\lb{mufz}
u^{i'}=\mf^{i'}(z^{\ti}_{\tk}),\qquad\quad
i'=1,\dots,\diu.
\end{equation}
\textup{(}We can do this on a neighborhood of a point where 
$\det\left(\dfrac{\partial}{\partial u^j}\big(\cS^{\da_i+1}(z^i)\big)\right)
\neq 0$.\textup{)}

Then one obtains an MT for equation~\eqref{sdde} as follows.
We introduce new variables $\fv_k^i$ for $i=1,\dots,\diu$ and $k\in\zz$.
For each $j\in\zz$, we define the operator $\cS^j$ on functions $h(\fv_k^i)$ 
by the usual rule 
$$
\cS^j\big(h(\fv_k^i)\big)=h(\cS^j(\fv_k^i)),\qquad\quad 
\cS^j(\fv_k^i)=\fv_{k+j}^i.
$$
That is, applying $\cS^j$ to a function $h=h(\fv_k^i)$, 
we replace $\fv_k^i$ by $\fv_{k+j}^i$ in $h$ for all $i,k$.

Using the function $\mf^{i'}(z^{\ti}_{\tk})$ from~\eqref{mufz}, 
we can consider the function $\mf^{i'}(\fv^{\ti}_{\tk})$.
Let $P^i(\fv_{\hk}^{\hi})$ be the function 
obtained from 
$Q^i(z^{i'}_{k'},u_\ta^{i'},\dots,u_\tb^{i'})$
by replacing $z^{i'}_{k'}$ with $\fv^{i'}_{k'}$ and $u_j^{i'}$ with 
$\cS^j\big(\mf^{i'}(\fv^{\ti}_{\tk})\big)$.
That is,
\begin{equation}
\lb{mpq}
P^i(\fv_{\hk}^{\hi})=
Q^i\big(\fv^{i'}_{k'},\cS^{\ta}\big(\mf^{i'}(\fv^{\ti}_{\tk})\big),\dots,
\cS^{\tb}\big(\mf^{i'}(\fv^{\ti}_{\tk})\big)\big),\qquad
i=1,\dots,\diu.
\end{equation}

We introduce the formula
\begin{equation}
\lb{mvtp}
v^i_t=P^i(\fv_{\hk}^{\hi}),\qquad\quad i=1,\dots,\diu,
\end{equation}
and regard~\eqref{mvtp} as a differential-difference equation 
for the variables~$\fv^1,\dots,\fv^\diu$.

Then the formulas $u^{i'}=\mf^{i'}(\fv^{\ti}_{\tk})$, $i'=1,\dots,\diu$, 
determine an MT from equation~\eqref{mvtp} to equation~\eqref{msdde}, 
where \eqref{msdde} is the component form of~\eqref{sdde}.
\end{theorem}

\begin{proof}
According to Remark~\ref{remhri}, one can regard $z^1,\dots,z^{\diu}$ 
as functions on an open dense subset of~$G$.
Since for any $k\in\zz$ we can apply the operator $\cS^k$ to such functions, 
we can consider $z^i_k=\cS^k(z^i)$ for all $k\in\zz$.

The functions $z^i_k$ may depend on the coordinates of $G$ and the variables $u_l^\ga$.
Similarly to Lemma~\ref{lemfi}, one shows that $z^i_k$, $k\in\zz$, $i=1,\dots,\diu$,  
are functionally independent.

According to Definition~\ref{defmtt} of MTs, 
to prove the theorem, we need to show that
\begin{multline}
\lb{mmteq}
\sum_{\ci,\ck}
\frac{\pd}{\pd\fv^{\ci}_{\ck}}\big(\mf^{i'}(\fv^{\ti}_{\tk})\big)\cdot 
\cS^{\ck}\big(P^{\ci}(\fv_{\hk}^{\hi})\big)=
\ff^{i'}\big(\cS^\ua\big(\mf^\ga(\fv^{\ti}_{\tk})\big),
\cS^{\ua+1}\big(\mf^\ga(\fv^{\ti}_{\tk})\big),\dots,
\cS^\ub\big(\mf^\ga(\fv^{\ti}_{\tk})\big)\big),\\
i'=1,\dots,\diu.
\end{multline}
Applying $D_t$ to equation~\eqref{mufz} 
and using the identity 
$D_t(u^{i'})=\ff^{i'}(u_\ua^\ga,u_{\ua+1}^\ga,\dots,u_\ub^\ga)$, one obtains
\begin{equation}
\lb{mffumf}
\ff^{i'}(u_\ua^\ga,u_{\ua+1}^\ga,\dots,u_\ub^\ga)
=\sum_{\ci,\ck}
\frac{\pd}{\pd z^{\ci}_{\ck}}\big(\mf^{i'}(z^{\ti}_{\tk})\big)\cdot D_t(z^{\ci}_{\ck}),
\qquad\quad
i'=1,\dots,\diu.
\end{equation}
Since $\cS\circ D_t=D_t\circ\cS$, we have
\begin{equation}
\lb{mdtzi}
D_t(z^{\ci}_{\ck})=D_t(\cS^{\ck}(z^{\ci}))=
\cS^{\ck}(D_t(z^{\ci}))=
\cS^{\ck}\big(Q^{\ci}(z^{i'}_{k'},u_\ta^{i'},\dots,u_\tb^{i'})\big).
\end{equation}
For any $j\in\zz$, applying $\cS^j$ to equation~\eqref{mufz}, one gets
\begin{equation}
\lb{mujfz}
u_j^{i'}=\cS^j(\mf^{i'}(z^{\ti}_{\tk})),\qquad\quad
i'=1,\dots,\diu,\qquad j\in\zz.
\end{equation}
Substituting~\eqref{mujfz} and~\eqref{mdtzi} in~\eqref{mffumf}, we obtain
\begin{multline}
\lb{mfffmf}
\ff^{i'}\big(\cS^\ua\big(\mf^\ga(z^{\ti}_{\tk})\big),
\cS^{\ua+1}\big(\mf^\ga(z^{\ti}_{\tk})\big),\dots,
\cS^\ub\big(\mf^\ga(z^{\ti}_{\tk})\big)\big)=\\
=\sum_{\ci,\ck}
\frac{\pd}{\pd z^{\ci}_{\ck}}\big(\mf^{i'}(z^{\ti}_{\tk})\big)\cdot 
\cS^{\ck}\big(
Q^{\ci}\big(z^{i'}_{k'},\cS^{\ta}\big(\mf^{i'}(z^{\ti}_{\tk})\big),\dots,
\cS^{\tb}\big(\mf^{i'}(z^{\ti}_{\tk})\big)\big)
\big),\qquad
i'=1,\dots,\diu.
\end{multline}
Since $z^i_k$, $k\in\zz$, $i=1,\dots,\diu$, are functionally independent,
the identity~\eqref{mfffmf} will remain valid if we replace $z^i_k$ by $\fv^i_k$ for all $i,k$.
Replacing $z^i_k$ by $\fv^i_k$ in~\eqref{mfffmf} and using~\eqref{mpq}, we get~\eqref{mmteq}.
\end{proof}

Examples of MTs in the case $\diu=2$ are constructed in Sections~\ref{sectoda},~\ref{secap}.

\begin{remark}
\lb{mnondeg}
In Remark~\ref{nondeg} we have discussed why~\eqref{sclc},~\eqref{pdu} can be regarded as 
non-degeneracy conditions. 
The same applies to~\eqref{msclc},~\eqref{mpdu}.
\end{remark}

\begin{remark}
\lb{remact}
Let $\wor{\sm}$ be the space $\fik^\sm$ without $0$.
Recall that the $(\sm-1)$-dimensional complex projective space $\prs^{\sm-1}$ 
can be identified with the set of one-dimensional subspaces in~$\fik^\sm$.
The standard action of the group $\GL_\sm(\fik)$ on $\wor{\sm}$
is transitive, which gives a transitive action of $\GL_\sm(\fik)$ on $\prs^{\sm-1}$.

For $i=1,\dots,\og$, let $M_i$ be either $\wor{\sm}$ or $\prs^{\sm-1}$. 
As has been shown above, one has the standard transitive action 
of $\GL_\sm(\fik)$ on $M_i$.
Set $\mnf=M_1\times\dots\times M_\og$.

Recall that $G$ is defined by~\eqref{ggr}.
Consider the case $k_1=\dots=k_\og=1$, 
so 
\begin{equation}
\lb{gk1}
G=\gq{\sm}{\cp_1}{1}\times\dots\times\gq{\sm}{\cp_\og}{1}.
\end{equation}
Using the isomorphism $\gq{\sm}{\cp_i}{1}\cong\GL_\sm(\fik)$ 
and the transitive action of $\GL_\sm(\fik)$ on~$M_i$ 
for $i=1,\dots,\og$, we obtain a transitive action 
of $G$ on $\mnf={M_1\times\dots\times M_\og}$.

Let $\pmm\in\mnf$.
Let $H\subset G$ be the stabilizer subgroup of the point~$\pmm$.
That is, 
\begin{equation}
\lb{gxx}
H=\{g\in G\ |\ g\pmm=\pmm\}.
\end{equation}
Then we have the analytic diffeomorphism 
$$
G/H\,\to\, \mnf=M_1\times\dots\times M_\og,\qquad 
gH\mapsto g\pmm,\qquad\forall\,g\in G.
$$ 
Note that $\codim H=\dim G/H=\dim M_1+\dots+\dim M_\og$, because 
$G/H$ is diffeomorphic to $\mnf=M_1\times\dots\times M_\og$.

Therefore, when we apply Theorems~\ref{thscal} or \ref{tharb} to such $G$ and $H$, 
we can replace $G/H$ by $\mnf$ 
and consider functions on open dense subsets of $\mnf$.
Examples of MTs arising from such $G$ and $H$ are given in the next sections.

Since $\gq{\sm}{\cp_i}{1}\cong\GL_\sm(\fik)$, the structure of the group~\eqref{gk1}
does not depend on $\cp_1,\dots,\cp_\og$.
However, the homomorphism $\rho\colon\gla\to G$ given by~\eqref{rhog},~\eqref{gprho} 
depends on $\cp_1,\dots,\cp_\og$.
Hence the subgroups $\rho(\hg{k})\subset G$ for $k\in\zp$ 
and the set of $\hg{k}$-left-invariant functions depend on $\cp_1,\dots,\cp_\og$ as well.
\end{remark}

\section{MTs for the Narita-Itoh-Bogoyavlensky lattice}
\lb{secbog}

For each $p\in\zsp$, the Narita-Itoh-Bogoyavlensky lattice~\cite{bogoy88,itoh87,narita82}
is the following differential-difference equation
\begin{equation}
\label{bog}
u_{t}=u\Big(\sum_{k=1}^pu_{k}-\sum_{k=1}^pu_{-k}\Big).
\end{equation}
It possesses the operator Lax pair~\cite{bogoy88,kmw}
$$
L=\cS+u \cS^{-p},\qquad\quad A=(L^{(p+1)})_{\ge 0},
$$
where $\ge 0$ means taking the terms with nonnegative power of $\cS$ in $L^{(p+1)}$.
The corresponding Lax equation $\pd_t(L)=[A,L]$ is equivalent to~\eqref{bog}.

This implies that equation~\eqref{bog}
is equivalent to the compatibility of the linear system
\begin{gather}
\lb{lphi}
L\phi=\lambda\phi,\\
\lb{phit}
\phi_t=A\phi,
\end{gather}
where $\phi=\phi(n,t)$ is a scalar function and $\la$ is a parameter.

\subsection{The case p=2}
\lb{subsp2}

For $p=2$ equation~\eqref{bog} reads
\begin{equation}
\label{bogp2}
u_t=u(u_2+u_1-u_{-1}-u_{-2}).
\end{equation}
In this case $L=\cS+u\cS^{-2}$, hence 
\begin{equation}
\lb{ap2}
A=(L^3)_{\ge 0}=
\big((\cS^2+(u_1+u)\cS^{-1}+uu_{-2}\cS^{-4})(\cS+u\cS^{-2})\big)_{\ge 0}
=\cS^3+u+u_1+u_2.
\end{equation}

To apply the theory described in Section~\ref{secgt}, 
we need to find a DLR of the form~\eqref{mu},~\eqref{lr}.
Recall that equation~\eqref{bogp2}  
is equivalent to the compatibility of system~\eqref{lphi},~\eqref{phit}.
To obtain a DLR of the form~\eqref{mu},~\eqref{lr}, 
we are going to rewrite system~\eqref{lphi},~\eqref{phit} in the form~\eqref{syspsi} 
for some vector $\Psi$ and matrices $M$, $\lt$.

For $L=\cS+u\cS^{-2}$ equation~\eqref{lphi} reads
\begin{equation}
\lb{csphi}
\cS(\phi)+u \cS^{-2}(\phi)=\lambda \phi.
\end{equation}
Set 
\begin{equation}
\lb{psip}
\psi^1=\cS^{-2}(\phi),\quad\qquad\psi^2=\cS^{-1}(\phi),\quad\qquad\psi^3=\phi.
\end{equation}
Then equation~\eqref{csphi} is equivalent to
\begin{equation}
\label{origM}
\cS
\begin{pmatrix}
 \psi^1 \\
 \psi^2 \\
 \psi^3 \\
\end{pmatrix}
=
\begin{pmatrix}
 0 & 1 & 0 \\
 0 & 0 & 1 \\
 -u & 0 & \lambda  \\
\end{pmatrix}
\begin{pmatrix}
\psi^1 \\
 \psi^2 \\
 \psi^3 \\
\end{pmatrix}.
\end{equation}
Set 
$$
\Psi=
\begin{pmatrix}
\psi^1 \\
 \psi^2 \\
 \psi^3 \\
\end{pmatrix},\qquad\quad
M(u,\la)=
\begin{pmatrix}
 0 & 1 & 0 \\
 0 & 0 & 1 \\
 -u & 0 & \lambda  \\
\end{pmatrix}.
$$
Then~\eqref{origM} reads $\cS(\Psi)=M(u,\la)\Psi$. 
Applying $\cS^{-1}$ to this equation, we obtain 
$$
\Psi=\cS^{-1}\big(M(u,\la)\big)\cS^{-1}(\Psi)=M(u_{-1},\la)\cS^{-1}(\Psi),
$$ 
which implies $\cS^{-1}(\Psi)=M(u_{-1},\la)^{-1}\Psi$, that is,
\begin{equation}
\lb{s1p}
\cS^{-1}
\begin{pmatrix}
 \psi_1 \\
 \psi_2 \\
 \psi_3 \\
\end{pmatrix}
=
\begin{pmatrix}
 0 & \frac{\lambda }{u_{-1}} & -\frac{1}{u_{-1}} \\
 1 & 0 & 0 \\
 0 & 1 & 0 \\
\end{pmatrix}
\begin{pmatrix}
 \psi_1 \\
 \psi_2 \\
 \psi_3 \\
\end{pmatrix}.
\end{equation}

According to~\eqref{ap2}, we have $A=\cS^3+u+u_1+u_2$, so equation~\eqref{phit} reads
\begin{equation}
\lb{phitp2}
\phi_t=\cS^3(\phi)+(u+u_1+u_2)\phi.
\end{equation}
Using~\eqref{psip},~\eqref{origM}, we can rewrite~\eqref{phitp2} as
\begin{equation}
\lb{p3t}
\psi^3_t=\cS^3(\psi^{3})+(u+u_1+u_2)\psi^{3}
=-\lambda^2u\psi^1-\lambda u_1\psi^2+(\lambda^3+u+u_1)\psi^3,
\end{equation}
where we have used \eqref{origM} to compute $\cS^3(\psi^3)$.
According to~\eqref{psip}, one has $\psi^{2}= S^{-1}(\psi^3)$ and $\psi^1=S^{-1}(\psi^2)$.
Combining this with~\eqref{s1p} and~\eqref{p3t}, we get 
\begin{gather}
\lb{p2t}
\begin{split}
\psi^{2}_t= S^{-1}(\psi^3_t)
&=-\lambda^2u_{-1}S^{-1}(\psi^1)-\lambda uS^{-1}(\psi^2)+(\lambda^3+u_{-1}+u)S^{-1}(\psi^3)=\\
&=-\lambda u\psi^{1}+(u_{-1}+u)\psi^{2}+\lambda^2 \psi^{3},
\end{split}\\
\lb{p1t}
\begin{split}
\psi^1_{t}=S^{-1}(\psi^2_t)
&=-\lambda u_{-1}S^{-1}(\psi^1)+(u_{-2}+u_{-1})S^{-1}(\psi^2)+\lambda^2S^{-1}(\psi^3)=\\
&=(u_{-1}+u_{-2})\psi^{1}+\lambda \psi^{3}.
\end{split}
\end{gather}
Equations~\eqref{p3t},~\eqref{p2t},~\eqref{p1t} can be written in matrix form as follows
\begin{equation}
\lb{b2tm}
\partial_t
\begin{pmatrix}
 \psi^1 \\
 \psi^2 \\
 \psi^3 \\
\end{pmatrix}
=
\begin{pmatrix}
 u_{-2}+u_{-1} & 0 & \lambda  \\
 -\lambda  u & u_{-1}+u & \lambda^2 \\
 -\lambda^2u  & -\lambda  u_1 & \lambda^3+u_1+u \\
\end{pmatrix}
\begin{pmatrix}
 \psi^1 \\
 \psi^2 \\
 \psi^3 \\
\end{pmatrix}.
\end{equation}

Since the compatibility of system~\eqref{origM},~\eqref{b2tm} is equivalent to equation~\eqref{bogp2},
we obtain the following DLR for~\eqref{bogp2}
\begin{equation}
\lb{bmu}
M(u,\la)=
\begin{pmatrix}
 0 & 1 & 0 \\
 0 & 0 & 1 \\
 -u & 0 & \lambda  \\
\end{pmatrix}
,\qquad\quad
\lt=
\begin{pmatrix}
 u_{-2}+u_{-1} & 0 & \lambda  \\
 -\lambda  u & u_{-1}+u & \lambda^2 \\
 -\lambda^2u  & -\lambda  u_1 & \lambda^3+u_1+u \\
\end{pmatrix}.
\end{equation}
Note that $M(u,\la)$ is not invertible for $u=0$. 
In agreement with Remark~\ref{reminv}, we assume $u\neq 0$.

Let us compute the groups $\hg{1}$, $\hg{2}$, $\hg{3}$ for this DLR.
In the course of computation, we will see that it is convenient 
to replace~\eqref{bmu} by a gauge-equivalent DLR,
in agreement with Remark~\ref{conjug}.

According to~\eqref{mumu}, the group $\hg{1}$ is generated by the matrix-functions 
\begin{equation}
\lb{bmh1}
M({\tilde u},\la)\cdot M(u,\la)^{-1}
=
\begin{pmatrix}
 1 & 0 & 0 \\
 0 & 1 & 0 \\
 0 & \lambda\big(1-\frac{{\tilde u}}{u}\big) & \frac{{\tilde u}}{u} \\
\end{pmatrix}.
\end{equation}
Note that the right-hand side of~\eqref{bmh1} can be simplified by the following conjugation
$$
\begin{pmatrix}
 1 & 0 & 0 \\
 0 & 1 & 0 \\
 0 & -\lambda  & 1 \\
\end{pmatrix}
\begin{pmatrix}
 1 & 0 & 0 \\
 0 & 1 & 0 \\
 0 & \lambda\big(1-\frac{{\tilde u}}{u}\big) & \frac{{\tilde u}}{u} \\
\end{pmatrix}
\begin{pmatrix}
 1 & 0 & 0 \\
 0 & 1 & 0 \\
 0 & -\lambda  & 1 \\
\end{pmatrix}^{-1}=
\begin{pmatrix}
 1 & 0 & 0 \\
 0 & 1 & 0 \\
 0 & 0 & \frac{{\tilde u}}{u} \\
\end{pmatrix}.
$$
Therefore, it makes sense to consider the matrix
$\mathbf{g}^1(\la)=
\begin{pmatrix}
 1 & 0 & 0 \\
 0 & 1 & 0 \\
 0 & -\lambda  & 1 \\
\end{pmatrix}$ 
and the following gauge-equivalent DLR
\begin{equation}
\notag
\hat M(u,\la)=\mathbf{g}^1(\la)\cdot M(u,\la)\cdot\mathbf{g}^1(\la)^{-1}
=
\begin{pmatrix}
 0 & 1 & 0 \\
 0 & \lambda  & 1 \\
 -u & 0 & 0 \\
\end{pmatrix},\qquad\quad
\hat\lt=\mathbf{g}^1(\la)\cdot\lt\cdot\mathbf{g}^1(\la)^{-1}.
\end{equation}
For this DLR,
the group $\hg{1}$ is generated by the matrix-functions 
\begin{equation}
\notag
\hat M({\tilde u},\la)\cdot\hat M(u,\la)^{-1}
=
\begin{pmatrix}
 1 & 0 & 0 \\
 0 & 1 & 0 \\
 0 & 0 & \frac{{\tilde u}}{u} \\
\end{pmatrix},
\end{equation}
so $\hg{1}$ consists of the constant matrix-functions 
$\begin{pmatrix}
 1 & 0 & 0 \\
 0 & 1 & 0 \\
 0 & 0 & a_1 \\
\end{pmatrix}$,  
where $a_1\in\fik$ is an arbitrary nonzero constant.

According to~\eqref{mhm}, the group $\hg{2}$ is generated by 
$h_1=
\begin{pmatrix}
 1 & 0 & 0 \\
 0 & 1 & 0 \\
 0 & 0 & a_1 \\
\end{pmatrix}
\in\hg{1}$ and 
\begin{equation}
\lb{bmh1m}
\hat M(u,\la)\cdot h_1\cdot\hat M(u,\la)^{-1}=
\begin{pmatrix}
 1 & 0 & 0 \\
 \lambda(1-a_1) & a_1 & 0 \\
 0 & 0 & 1 \\
\end{pmatrix}
\end{equation}
for all nonzero $a_1\in\fik$.
The right-hand side of~\eqref{bmh1m} can be simplified by the following conjugation
$$
\begin{pmatrix}
 1 & 0 & 0 \\
 -\lambda  & 1 & 0 \\
 0 & 0 & 1 \\
\end{pmatrix}
\begin{pmatrix}
 1 & 0 & 0 \\
 \lambda(1-a_1) & a_1 & 0 \\
 0 & 0 & 1 \\
\end{pmatrix}
\begin{pmatrix}
 1 & 0 & 0 \\
 -\lambda  & 1 & 0 \\
 0 & 0 & 1 \\
\end{pmatrix}^{-1}=
\begin{pmatrix}
 1 & 0 & 0 \\
 0 & a_1 & 0 \\
 0 & 0 & 1 \\
\end{pmatrix}.
$$
Hence, it makes sense to consider the matrix
$\mathbf{g}^2(\la)=
\begin{pmatrix}
 1 & 0 & 0 \\
 -\lambda  & 1 & 0 \\
 0 & 0 & 1 \\
\end{pmatrix}$ 
and the following gauge-equivalent DLR
\begin{gather}
\lb{chem}
\check M(u,\la)=\mathbf{g}^2(\la)\cdot\hat M(u,\la)\cdot\mathbf{g}^2(\la)^{-1}
=
\begin{pmatrix}
 \lambda  & 1 & 0 \\
 0 & 0 & 1 \\
 -u & 0 & 0 \\
\end{pmatrix},\\
\lb{cheu}
\check\lt=\mathbf{g}^2(\la)\cdot\hat\lt\cdot\mathbf{g}^2(\la)^{-1}=
\begin{pmatrix}
 \lambda^3+u_{-2}+u_{-1} & \lambda^2 & \lambda  \\
 -\lambda  u_{-2} & u_{-1}+u & 0 \\
 -\lambda^2 u_{-1} & -\lambda  u_{-1} & u+u_{1} \\
\end{pmatrix}.
\end{gather}
For this DLR, the groups $\hg{1}$, $\hg{2}$ are as follows
\begin{gather}
\lb{bhg1}
\hg{1}=\left.\left\{
\begin{pmatrix}
 1 & 0 & 0 \\
 0 & 1 & 0 \\
 0 & 0 & a_1 \\
\end{pmatrix}\ \right|\ 
a_1\in\fik,\ a_1\neq 0
\right\},\\
\lb{bhg2}
\hg{2}=\left.\left\{
\begin{pmatrix}
 1 & 0 & 0 \\
 0 & a_2 & 0 \\
 0 & 0 & a_1 \\
\end{pmatrix}\ \right|\ 
a_1,a_2\in\fik,\ a_1,a_2\neq 0
\right\}.
\end{gather}

According to~\eqref{mhm}, the group $\hg{3}$ is generated by 
$h_2=
\begin{pmatrix}
 1 & 0 & 0 \\
 0 & a_2 & 0 \\
 0 & 0 & a_1 \\
\end{pmatrix}
\in\hg{2}$ and 
\begin{equation}
\notag
\check M(u,\la)\cdot h_2\cdot\check M(u,\la)^{-1}=
\begin{pmatrix}
 a_2 & 0 & \frac{(a_2-1) \lambda }{u} \\
 0 & a_1 & 0 \\
 0 & 0 & 1 \\
\end{pmatrix}
\end{equation}
for all nonzero $a_1,a_2,u\in\fik$.
Therefore, the group $\hg{3}$ for the DLR~\eqref{chem}, \eqref{cheu}
consists of the matrix-functions 
$\begin{pmatrix}
 a_3 & 0 & a_4\lambda\\
 0 & a_2 & 0 \\
 0 & 0 & a_1 \\
\end{pmatrix}$
for all $a_1,a_2,a_3,a_4\in\fik$, $a_1,a_2,a_3\neq 0$.

Let us construct some MTs for equation~\eqref{bogp2},
using Theorem~\ref{thscal} and Remark~\ref{remact}.
To clarify the construction, 
we begin with simple examples (Examples~\ref{bc3} and~\ref{bcp2}), 
which give some known MTs. Example~\ref{bcpcp} is more interesting 
and gives an MT which seems to be new.

Following Remark~\ref{remact}, in Examples~\ref{bc3},~\ref{bcp2},~\ref{bcpcp} 
we consider a group $G$ of the form~\eqref{gk1} and a transitive action 
of~$G$ on a manifold $\mnf$.
We take $H\subset G$ to be the stabilizer subgroup of a point~$\pmm\in\mnf$, 
so $H$ is given by~\eqref{gxx}. 
Then $G/H\cong\mnf$.
Since the results do not depend on the choice of $\pmm\in\mnf$, 
in Examples~\ref{bc3},~\ref{bcp2},~\ref{bcpcp} 
we do not mention $\pmm$ explicitly.

\begin{example}
\lb{bc3}

Using the notation from Section~\ref{sgalr}, 
consider the case $\sm=3$, $\og=1$, $k_1=1$, $\cp_1\in\fik$. 
Then $G=\gq{3}{\cp_1}{1}\cong\GL_3(\fik)$.

Consider the space $\fik^3$ with coordinates $\al^1$, $\al^2$, $\al^3$
and the standard transitive action of the group $G\cong\GL_3(\fik)$ on the manifold
$\mnf=\wor{3}$.

According to Remark~\ref{remact}, to apply Theorem~\ref{thscal}, 
we need to find an $\hg{2}$-left-invariant function~$z$ on an open dense subset 
of $\wor{3}$ such that conditions~\eqref{sclc},~\eqref{pdu} hold for $G/H\cong\wor{3}$ and $r=\dim G/H=3$.

Since we are using the DLR~\eqref{chem},~\eqref{cheu} 
and the standard action of $G\cong\GL_3(\fik)$ on $\wor{3}$, 
one has the formulas
\begin{align}
\lb{sal}
\cS
\begin{pmatrix}
 \al^1 \\
 \al^2 \\
 \al^3 \\
\end{pmatrix}
&=
\begin{pmatrix}
 \cp_1  & 1 & 0 \\
 0 & 0 & 1 \\
 -u & 0 & 0 \\
\end{pmatrix}
\begin{pmatrix}
 \al^1 \\
 \al^2 \\
 \al^3 \\
\end{pmatrix},\\
\lb{dtal}
D_t
\begin{pmatrix}
 \al^1 \\
 \al^2 \\
 \al^3 \\
\end{pmatrix}
&=
\begin{pmatrix}
 \cp_1^3+u_{-2}+u_{-1} & \cp_1^2 & \cp_1  \\
 -\cp_1  u_{-2} & u_{-1}+u & 0 \\
 -\cp_1^2 u_{-1} & -\cp_1 u_{-1} & u+u_{1} \\
\end{pmatrix}
\begin{pmatrix}
 \al^1 \\
 \al^2 \\
 \al^3 \\
\end{pmatrix}.
\end{align}

In the right-hand side of~\eqref{sal} we use the matrix
$\check M(u,\la)$ given by~\eqref{chem}, but we substitute $\cp_1$ in place of $\la$, 
for the following reason.
In the general definition~\eqref{csw} of~$\cS$ we work over the algebra
$\fik[\la]/((\la-\cp_p)^{k_p})$, where one has $(\la-\cp_p)^{k_p}=0$.
In the present example we have $k_1=1$, so we need to impose the relation $\la-\cp_1=0$.

The right-hand side of~\eqref{dtal} is obtained in the same way from the matrix 
$\check\lt$ given by~\eqref{cheu}.

As the group $\hg{2}$ is of the form~\eqref{bhg2}, 
the function $z=\al^1$ is $\hg{2}$-left-invariant.
Using formula~\eqref{sal} and the notation $z_k=\cS^k(z)$ for $k\in\zz$, we obtain
\begin{gather}
\lb{b2z1}
z_0=z=\al^1,\qquad\quad
z_1=\cS(z)=\cS(\al^1)=\cp_1\al^1+\al^2,\\
\lb{b2z2}
z_2=\cS(z_1)=\cp_1\cS(\al^1)+\cS(\al^2)
=\cp_1^2\al^1+\cp_1\al^2+\al^3,\\
\lb{b2z3}
z_3=\cS^3(z)=\cS(z_2)=\cp_1^2\cS(\al^1)+\cp_1\cS(\al^2)+\cS(\al^3)=
(\cp_1^3-u)\al^1+\cp_1^2\al^2+\cp_1\al^3.
\end{gather}

Recall that in this example we have $r=3$.
From~\eqref{b2z1},~\eqref{b2z2} we see that the functions $z$, $\cS(z)$, $\cS^2(z)$ 
form a system of coordinates on the manifold~$\wor{3}$, so condition~\eqref{sclc} is valid.
Equation~\eqref{b2z3} shows that condition~\eqref{pdu} is valid as well.

Therefore, using Theorem~\ref{thscal}, we obtain an MT as follows.
Equations~\eqref{b2z1}, \eqref{b2z2}, \eqref{b2z3} imply $z_3=-uz_0+\cp_1 z_2$, which yields
\begin{equation}
\lb{bu0z}
u=\frac{\cp_1 z_2-z_3}{z_0}.
\end{equation}
Using~\eqref{dtal},~\eqref{b2z1},~\eqref{b2z2}, we get
\begin{equation}
\lb{bdtz}
D_t(z)=D_t(\al^1)=
(\cp_1^3+u_{-2}+u_{-1})\al^{1}+\cp_1^2\al^2+\cp_1\al^3=
(u_{-2}+u_{-1})z_0+\cp_1 z_2.
\end{equation}
Equation~\eqref{bu0z} implies 
$$
u_{-1}=\cS^{-1}\left(\frac{\cp_1 z_2-z_3}{z_0}\right)=
\frac{\cp_1 z_1-z_2}{z_{-1}},\qquad\quad
u_{-2}=\cS^{-2}\left(\frac{\cp_1 z_2-z_3}{z_0}\right)=
\frac{\cp_1 z_0-z_1}{z_{-2}}.
$$
Substituting this in~\eqref{bdtz}, one obtains
\begin{equation}
\lb{bdtzz}
D_t(z)=\left(\frac{\cp_1 z_0-z_1}{z_{-2}}+\frac{\cp_1 z_1-z_2}{z_{-1}}\right)z_0+\cp_1 z_2.
\end{equation}
According to Theorem~\ref{thscal}, to obtain an MT, 
we need to replace $z_k$ by $v_k$ for all $k\in\zz$ in~\eqref{bu0z},~\eqref{bdtzz}.
Here and below we use also the identification $v_0=v$.

Thus we get the following result. For any $\cp_1\in\fik$, 
the formula 
\begin{equation}
\lb{ucv23v}
u=\frac{\cp_1 v_2-v_3}{v}
\end{equation}
determines an MT from the equation 
\begin{equation}
\lb{vtcvp2}
v_t=\left(\frac{\cp_1 v-v_1}{v_{-2}}+\frac{\cp_1 v_1-v_2}{v_{-1}}\right)v+\cp_1 v_2.
\end{equation}
to equation~\eqref{bogp2}.
As is shown in Remark~\ref{remmt1b} below, this MT can also be obtained 
immediately from system~\eqref{lphi},~\eqref{phit}.
We have included this simple example, in order to demonstrate how
the method described in Theorem~\ref{thscal} and Remark~\ref{remact} works.
\end{example}

\begin{remark}
\lb{remmt1b}
The MT~\eqref{ucv23v} can be obtained immediately from system~\eqref{lphi},~\eqref{phit} 
as follows.
Recall that in this subsection we consider the case $p=2$.
Since for $p=2$ we have $L=\cS+u \cS^{-2}$, 
the equation $L\phi=\lambda\phi$ is equivalent to
\begin{equation}
\lb{ulphi}
u=\frac{\lambda\phi-\cS(\phi)}{\cS^{-2}(\phi)}.
\end{equation}
Setting $v=\cS^{-2}(\phi)$ and replacing $\lambda$ by $\cp_1$ in~\eqref{ulphi}, 
one gets~\eqref{ucv23v}.
Substituting~\eqref{ucv23v} and $\phi=\cS^2(v)$ in~\eqref{phit} for $p=2$, 
one obtains an equation equivalent to~\eqref{vtcvp2}.
\end{remark}

\begin{example}
\lb{bcp2}
Consider the case 
$$
\sm=3,\quad \og=1,\quad k_1=1,\quad
G=\gq{3}{\cp_1}{1}\cong\GL_3(\fik),\quad \cp_1\in\fik,
$$
and the standard transitive action of the group $G\cong\GL_3(\fik)$ 
on the $2$-dimensional complex projective space $\prs^{2}$.
So $\mnf=\prs^{2}$ in this example.

According to Remark~\ref{remact}, to apply Theorem~\ref{thscal}, 
one needs to find an $\hg{1}$-left-invariant function~$z$ on an open dense subset 
of $\prs^{2}$ such that conditions~\eqref{sclc},~\eqref{pdu} 
hold for $G/H\cong\prs^{2}$ and $r=\dim G/H=2$.

Recall that in Example~\ref{bc3} we have considered the space  
$\fik^3$ with coordinates $\al^1$, $\al^2$, $\al^3$.
Now we regard $\al^1$, $\al^2$, $\al^3$ as homogeneous coordinates 
for the projective space $\prs^{2}$.

Then $y^1=\al^2/\al^1$, $\,\,y^2=\al^3/\al^1$ are affine coordinates on 
the open dense subset 
$$
\mathbb{V}=\big\{(\al^1:\al^2:\al^3)\in\prs^{2}\ \big|\ \al^1\neq 0\big\}\,\subset\,\prs^{2}.
$$

As the group $\hg{1}$ is of the form~\eqref{bhg1}, 
the function $z=y^1=\al^2/\al^1$ is $\hg{1}$-left-invariant.
Using formula~\eqref{sal} and the notation $z_k=\cS^k(z)$ for $k\in\zz$, we obtain
\begin{gather}
\lb{bcp2z1}
z_0=z=y^1=\frac{\al^2}{\al^1},\qquad\quad
z_1=\cS(z)=\cS\Big(\frac{\al^2}{\al^1}\Big)=
\frac{\al^3}{\cp_1\al^1+\al^2}=\frac{y^2}{\cp_1+y^1},\\
\lb{bcp2z2}
z_2=\cS^2(z)=\cS(z_1)=\frac{\cS(\al^3)}{\cp_1\cS(\al^1)+\cS(\al^2)}=
\frac{-u\al^1}{\cp_1(\cp_1\al^1+\al^2)+\al^3}=
\frac{-u}{(\cp_1+z_0)(\cp_1+z_1)}.
\end{gather}

Recall that in this example we have $r=2$.
Since the differentials of the functions $z$ and $\cS(z)$ 
are linearly independent, 
condition~\eqref{sclc} is valid.
Equation~\eqref{bcp2z2} shows that condition~\eqref{pdu} is valid as well.
Therefore, using Theorem~\ref{thscal}, we obtain an MT as follows.
From~\eqref{bcp2z2} one gets
\begin{equation}
\lb{bcpu0z}
u=-z_2(\cp_1+z_0)(\cp_1+z_1).
\end{equation}
Using~\eqref{dtal} and~\eqref{bcp2z1}, we obtain
\begin{multline}
\lb{bcpdtz}
D_t(z)=D_t\Big(\frac{\al^2}{\al^1}\Big)
=\frac{D_t(\al^2)\al^1-\al^2D_t(\al^1)}{(\al^1)^2}=\\
=\frac{\big(-\cp_1 u_{-2}\al^1+(u_{-1}+u)\al^2\big)\al^1
-\al^2\big((\cp_1^3+u_{-2}+u_{-1})\al^1+\cp_1^2\al^2+\cp_1\al^3\big)}{(\al^1)^2}=\\
=uz_0-u_{-2}(\cp_1+z_0)-z_0\cp_1(\cp_1+z_0)(\cp_1+z_1).
\end{multline}
Equation~\eqref{bcpu0z} yields 
$u_{-2}=\cS^{-2}\big(-z_2(\cp_1+z_0)(\cp_1+z_1)\big)=
-z_0(\cp_1+z_{-2})(\cp_1+z_{-1})$.
Substituting this and~\eqref{bcpu0z} in~\eqref{bcpdtz}, we get
\begin{equation}
\lb{bcpzt}
D_t(z)=
z_0(\cp_1+z_0)\big((\cp_1+z_{-2})(\cp_1+z_{-1})-(\cp_1+z_{1})(\cp_1+z_{2})\big).
\end{equation}
According to Theorem~\ref{thscal}, to obtain an MT, 
we need to replace $z_k$ by $v_k$ for all $k\in\zz$ in~\eqref{bcpu0z},~\eqref{bcpzt}.
Thus we get the following result. For any $\cp_1\in\fik$, 
the formula 
\begin{equation}
\lb{uv2cvc}
u=-v_2(\cp_1+v)(\cp_1+v_1)
\end{equation}
determines an MT from the equation 
\begin{equation}
\lb{vcvcv2}
v_t=v(\cp_1+v)\big((\cp_1+v_{-2})(\cp_1+v_{-1})-(\cp_1+v_{1})(\cp_1+v_{2})\big)
\end{equation}
to equation~\eqref{bogp2}. 
This MT is equivalent to a known MT.
Indeed, set $\tilde v=v+\cp_1$. 
Then formulas~\eqref{uv2cvc},~\eqref{vcvcv2} can be written as
\begin{gather}
\lb{utvpcp}
u=-(\tilde{v}_2-\cp_1)\tilde v\tilde{v}_1,\\
\lb{tvttv}
\tilde v_t=(\tilde v-\cp_1)\tilde v
(\tilde{v}_{-2}\tilde{v}_{-1}-\tilde{v}_{1}\tilde{v}_{2}).
\end{gather}
As Yu.~B.~Suris told us, 
equation~\eqref{tvttv} and the MT~\eqref{utvpcp} are well known.
\end{example}

\begin{example}
\lb{bcpcp}

Now consider the case 
$$
\sm=3,\ \og=2,\ 
k_1=k_2=1,\,\  
G=\gq{3}{\cp_1}{1}\times\gq{3}{\cp_2}{1}\cong\GL_3(\fik)\times\GL_3(\fik),
\,\ 
\cp_1,\cp_2\in\fik.
$$
Using the standard transitive action of $\GL_3(\fik)$ on $\prs^{2}$,
we obtain a transitive action of the group $G\cong\GL_3(\fik)\times\GL_3(\fik)$ 
on $\mnf=\prs^{2}\times\prs^{2}$.

According to Remark~\ref{remact}, to apply Theorem~\ref{thscal}, 
we need to find an $\hg{3}$-left-invariant function~$z$ on an open dense subset 
of $\prs^{2}\times\prs^{2}$ such that conditions~\eqref{sclc},~\eqref{pdu} hold for 
$G/H\cong\prs^{2}\times\prs^{2}$ and $r=\dim G/H=4$.

As has been said in Example~\ref{bcp2}, 
we regard $\al^1$, $\al^2$, $\al^3$ as homogeneous coordinates 
for the $2$-dimensional projective space $\prs^{2}$.
We introduce also homogeneous coordinates $\be^1$, $\be^2$, $\be^3$ 
for another copy of~$\prs^{2}$. Similarly to~\eqref{sal},~\eqref{dtal}, one has the formulas
\begin{align}
\lb{sbe}
\cS
\begin{pmatrix}
 \be^1 \\
 \be^2 \\
 \be^3 \\
\end{pmatrix}
&=
\begin{pmatrix}
 \cp_1  & 1 & 0 \\
 0 & 0 & 1 \\
 -u & 0 & 0 \\
\end{pmatrix}
\begin{pmatrix}
 \be^1 \\
 \be^2 \\
 \be^3 \\
\end{pmatrix},\\
\lb{dtbe}
D_t
\begin{pmatrix}
 \be^1 \\
 \be^2 \\
 \be^3 \\
\end{pmatrix}
&=
\begin{pmatrix}
 \cp_1^3+u_{-2}+u_{-1} & \cp_1^2 & \cp_1  \\
 -\cp_1  u_{-2} & u_{-1}+u & 0 \\
 -\cp_1^2 u_{-1} & -\cp_1 u_{-1} & u+u_{1} \\
\end{pmatrix}
\begin{pmatrix}
 \be^1 \\
 \be^2 \\
 \be^3 \\
\end{pmatrix}.
\end{align}

Recall that, according to Section~\ref{sgalr},
the group $\gla$ consists of \mbox{$\GL_\sm(\fik)$-valued} functions of~$\la$, 
and in the present example we have $\sm=3$.

As has been computed above, the subgroup $\hg{3}\subset\gla$ 
for the DLR~\eqref{chem}, \eqref{cheu} 
consists of the matrix-functions 
$\begin{pmatrix}
 a_3 & 0 & a_4\lambda\\
 0 & a_2 & 0 \\
 0 & 0 & a_1 \\
\end{pmatrix}
\in\gla$
for all $a_1,a_2,a_3,a_4\in\fik$, $a_1,a_2,a_3\neq 0$.

When we speak about $\hg{3}$-left-invariant functions 
on open subsets of $\prs^{2}\times\prs^{2}$, we need to consider 
the action of $\hg{3}$ on $\prs^{2}\times\prs^{2}$ that is determined 
by the embedding $\rho(\hg{3})\subset G$ and the action of $G$ on $\prs^{2}\times\prs^{2}$, 
where the homomorphism $\rho\colon\gla\to G$ is defined by~\eqref{rhog},~\eqref{gprho}.

According to~\eqref{rhog},~\eqref{gprho}, one has
$$
\rho\left(
\begin{pmatrix}
 a_3 & 0 & a_4\lambda\\
 0 & a_2 & 0 \\
 0 & 0 & a_1 \\
\end{pmatrix}
\right)
=
\left(
\begin{pmatrix}
 a_3 & 0 & a_4\cp_1\\
 0 & a_2 & 0 \\
 0 & 0 & a_1 \\
\end{pmatrix},
\begin{pmatrix}
 a_3 & 0 & a_4\cp_2\\
 0 & a_2 & 0 \\
 0 & 0 & a_1 \\
\end{pmatrix}
\right)\in\GL_3(\fik)\times\GL_3(\fik)\cong G.
$$
Therefore, the action of $\hg{3}$ on $\prs^{2}\times\prs^{2}$ is as follows
\begin{equation}
\notag
\begin{pmatrix}
 a_3 & 0 & a_4\lambda\\
 0 & a_2 & 0 \\
 0 & 0 & a_1 \\
\end{pmatrix}
\cdot \left(
\begin{pmatrix}
\al^1 \\
\al^2 \\
\al^3 \\
\end{pmatrix},
\begin{pmatrix}
\be^1 \\
\be^2 \\
\be^3 \\
\end{pmatrix}
\right)=
\left(
\begin{pmatrix}
a_3\al^1+a_4\cp_1\al^3\\
a_2\al^2 \\
a_1\al^3 \\
\end{pmatrix},
\begin{pmatrix}
a_3\be^1+a_4\cp_2\be^3\\
a_2\be^2\\
a_1\be^3\\
\end{pmatrix}
\right).
\end{equation}
This formula implies that the function $z=\dfrac{\al^2\be^3}{\al^3\be^2}$ is 
$\hg{3}$-left-invariant.

Consider the affine coordinates
\begin{equation}
\lb{qqqq}
q^1=\frac{\al^1}{\al^3},\qquad
q^2=\frac{\al^2}{\al^3},\qquad
\hat q^1=\frac{\be^1}{\be^3},\qquad
\hat q^2=\frac{\be^2}{\be^3}
\end{equation}
on the open dense subset $\mathbb{W}\subset\prs^{2}\times\prs^{2}$ where $\al^3,\be^3\neq 0$.

Using~\eqref{sal},~\eqref{sbe},~\eqref{qqqq}, and the formula 
$z=\dfrac{\al^2\be^3}{\al^3\be^2}=\dfrac{q^2}{\hat q^2}$, 
one can express the functions 
$\cS^k(z)$, $k=0,1,2,3$, in terms of $q^1$, $q^2$, $\hat q^1$, $\hat q^2$.
These expressions show that condition~\eqref{sclc} is valid.

Set $z_k=\cS^k(z)$ for $k\in\zz$.
Similarly to the previous examples, one obtains a formula of the form 
$z_4=\cS^{4}(z)=\fe(z_0,z_1,z_{2},z_{3},u)$ with some rational function $\fe$, 
which depends nontrivially on $u$, so condition~\eqref{pdu} is valid as well.
From the equation $z_4=\fe(z_0,z_1,z_2,z_3,u)$ 
one can express $u$ in terms of $z_0$, $z_1$, $z_2$, $z_3$, $z_4$ as follows
\begin{equation}
\lb{bucpcpz}
u=\frac{z_2 z_3(\cp_1 z_4-\cp_2)(\cp_2 z_0 z_1-\cp_1)
(\cp_2 z_1 z_2-\cp_1)}{(z_0 z_1 z_2-1)(z_1 z_2 z_3-1)(z_2 z_3 z_4-1)}.
\end{equation}

Similarly to the computation of $D_t(z)$ in Example~\ref{bcp2}, 
using formulas~\eqref{dtal}, \eqref{dtbe}, \eqref{bucpcpz}, one can obtain the following 
\begin{equation}
\lb{cpcpzt}
D_t(z)=\frac{z_0(z_{-2} z_{-1}-z_1 z_2)(\cp_2-\cp_1 z_0)(\cp_2 z_0 z_{-1}-\cp_1)(\cp_2 z_0 z_1-\cp_1)}{(z_0 z_{-2} z_{-1}-1)(z_0 z_{-1} z_1-1)(z_0 z_1 z_2-1)}.
\end{equation}

According to Theorem~\ref{thscal}, to obtain an MT, 
we need to replace $z_k$ by $v_k$ for all $k\in\zz$ in~\eqref{bucpcpz},~\eqref{cpcpzt}, 
which gives the following result.
For any $\cp_1,\cp_2\in\fik$, 
the formula 
\begin{equation}
\lb{uv2v3}
u=\frac{v_2 v_3(\cp_1 v_4-\cp_2)(\cp_2 v v_1-\cp_1)
(\cp_2 v_1 v_2-\cp_1)}{(v v_1 v_2-1)(v_1 v_2 v_3-1)(v_2 v_3 v_4-1)}.
\end{equation}
determines an MT from the equation 
\begin{equation}
\lb{vtvv2}
v_t=
\frac{v(v_{-2} v_{-1}-v_1 v_2)(\cp_2-\cp_1 v)(\cp_2 v v_{-1}-\cp_1)(\cp_2 v v_1-\cp_1)}{(v v_{-2} v_{-1}-1)(v v_{-1} v_1-1)(v v_1 v_2-1)}
\end{equation}
to equation~\eqref{bogp2}.

Equation~\eqref{vtvv2} and the MT~\eqref{uv2v3} with arbitrary $\cp_1,\cp_2\in\fik$
seem to be new.

In the case $\cp_2=0$ equation~\eqref{vtvv2} is equivalent 
(via scaling)
to a particular case of equation~(17.8.24) from~\cite[Section 17]{suris03}
and to equation~(43a) from~\cite[Section 4]{mx14}.
The MT~\eqref{uv2v3} for $\cp_2=0$ can be extracted from the results 
of~\cite[Section 17]{suris03} and~\cite[Section 4]{mx14} as well.  

In the case $\cp_1=\cp_2\neq 0$ equation~\eqref{vtvv2} 
can be transformed (by scaling of~$t$ and a linear-fractional transformation of~$v$)
to equation~(3.4) from the arxiv version of~\cite{levi14}.
Note that equation~(3.3b) from the journal version of~\cite{levi14} is supposed 
to be the same as equation~(3.4) from the arxiv version of~\cite{levi14}, 
but there is a misprint in equation~(3.3b) in the journal version of~\cite{levi14}.

Recall that the MT~\eqref{uv2v3} has been constructed from the DLR~\eqref{bmu}.
Note that the order of the MT~\eqref{uv2v3} is higher than the size 
of the matrices in the DLR~\eqref{bmu}.
Indeed, the MT~\eqref{uv2v3} is of order~$4$, while the matrices~\eqref{bmu} are of size~$3$. 

Recall that in this example we consider the case $p=2$.
See Section~\ref{secparb} for a discussion of an analogue of~\eqref{vtvv2},~\eqref{uv2v3} 
for arbitrary $p\in\zsp$ and its relations with some formulas of Yu.~B.~Suris.

\end{example}

\subsection{The case of arbitrary $p$}
\lb{secparb}

Recall that we have obtained the $3\times 3$ matrix DLR~\eqref{bmu} 
for equation~\eqref{bogp2} from system~\eqref{lphi},~\eqref{phit} for $p=2$.
Analogously, one can obtain a $(p+1)\times(p+1)$ matrix DLR 
for equation~\eqref{bog} from system~\eqref{lphi},~\eqref{phit} for arbitrary $p\in\zsp$.

Then one can construct MTs for equation~\eqref{bog}, 
using this DLR, Theorem~\ref{thscal}, and Remark~\ref{remact}.
An MT constructed in this way is discussed below.

In Example~\ref{bcpcp} we have studied the case 
\begin{gather*}
p=2,\qquad\sm=3,\qquad \og=2,\qquad k_1=k_2=1,\\
G=\gq{3}{\cp_1}{1}\times\gq{3}{\cp_2}{1}\cong\GL_3(\fik)\times\GL_3(\fik),
\qquad \cp_1,\cp_2\in\fik.
\end{gather*}
Using the standard transitive action 
of ${G\cong\GL_3(\fik)\times\GL_3(\fik)}$ on $\prs^{2}\times\prs^{2}$, 
in Example~\ref{bcpcp} we have obtained the MT~\eqref{uv2v3},~\eqref{vtvv2}.

Let us generalize this to the case of arbitrary $p\in\zsp$.
Since the corresponding computations are very similar to those 
in Section~\ref{subsp2}, we present only the final result.
Let 
\begin{gather*}
p\in\zsp,\qquad\sm=p+1,\qquad \og=2,\qquad k_1=k_2=1,\\
G=\gq{p+1}{\cp_1}{1}\times\gq{p+1}{\cp_2}{1}\cong
\GL_{p+1}(\fik)\times\GL_{p+1}(\fik),
\qquad \cp_1,\cp_2\in\fik.
\end{gather*}
Similarly to Example~\ref{bcpcp}, using the standard transitive action 
of ${G\cong\GL_{p+1}(\fik)\times\GL_{p+1}(\fik)}$ on $\prs^{p}\times\prs^{p}$, 
one can obtain the following.
For any $\cp_1,\cp_2\in\fik$, the formula
\begin{equation}
\lb{ucpvcp}
u=\frac{(\cp_1 v_{2p}-\cp_2)\big(\prod_{j=p}^{2p-1}v_j\big)%
\prod_{i=0}^{p-1}\Big(\cp_2\big(\prod_{j=i}^{i+p-1}v_j\big)-\cp_1\Big)}%
{\prod_{i=0}^p\big(-1+\prod_{j=i}^{i+p}v_j\big)}
\end{equation}
determines an MT from the equation
\begin{equation}
\lb{vtvcpcp}
v_t=\frac{v(\cp_2-\cp_1 v)\big(\prod_{i=1}^p v_{-i}-\prod_{i=1}^p v_{i}\big)%
\prod_{i=0}^{p-1}\Big(\cp_2\big(\prod_{j=i+1-p}^i v_j\big)-\cp_1\Big)}%
{\prod_{i=0}^p\big(-1+\prod_{j=i-p}^i v_j\big)}
\end{equation}
to equation~\eqref{bog}.
Formulas~\eqref{ucpvcp},~\eqref{vtvcpcp} are a generalization 
of~\eqref{uv2v3},~\eqref{vtvv2} to the case of arbitrary $p\in\zsp$.
The fact that~\eqref{ucpvcp} is indeed an MT from~\eqref{vtvcpcp} to~\eqref{bog}
can also be checked by a straightforward computation.

Equation~\eqref{vtvcpcp} and the MT~\eqref{ucpvcp} with arbitrary $\cp_1,\cp_2\in\fik$
seem to be new.
In the case $\cp_2=0$ equation~\eqref{vtvcpcp} is equivalent 
to equation~(17.8.24) from~\cite[Section 17]{suris03} via scaling of the variables.

We showed~\eqref{ucpvcp},~\eqref{vtvcpcp} to Yu.~B.~Suris, and he told us that 
this MT can be written in a more symmetric form as follows.
Let $\al,\beta$ be complex parameters.
Consider the following differential-difference equations for scalar functions 
$w=w(n,t)$ and $v=v(n,t)$
\begin{gather}
\lb{wtw}
w_t=w(\al+\beta w)\Big(\prod_{i=1}^pw_{i}-\prod_{i=1}^pw_{-i}\Big),\\
\lb{vtvalbe}
v_t=v(\al+\beta v)\Big(\prod_{i=1}^pv_{i}-\prod_{i=1}^pv_{-i}\Big)
\frac{\prod_{j=1}^p(1+\al\prod_{i=1}^pv_{j-i})}{\prod_{j=0}^p(1-\beta\prod_{i=0}^pv_{j-i})}.
\end{gather}
Equation~\eqref{wtw} is connected with~\eqref{bog} by the MT
\begin{equation}
\lb{mtuw}
u=(\al+\beta w_p)\prod_{i=0}^{p-1}w_i.
\end{equation}
It is easy to check that equations~\eqref{wtw},~\eqref{vtvalbe} are connected by the MT
\begin{equation}
\lb{wvalbe}
w=v
\frac{1+\al\prod_{i=1}^pv_{-i}}{1-\beta\prod_{i=0}^pv_{-i}}.
\end{equation}
In the case $\beta\neq 0$, equation~\eqref{vtvalbe} 
is equivalent to~\eqref{vtvcpcp} via scaling,  
and the composition of the MTs~\eqref{mtuw},~\eqref{wvalbe} 
is equivalent (up to a shift and scaling) to the MT~\eqref{ucpvcp}.

As Yu.~B.~Suris told us, equation~\eqref{wtw} and the MT~\eqref{mtuw} are well known.
Formulas \eqref{vtvalbe}, \eqref{wvalbe} with arbitrary $\alpha,\beta\in\fik$ 
seem to be new.
In the cases $(\alpha,\beta)=(1,0)$ and $(\alpha,\beta)=(0,1)$   
formulas equivalent to~\eqref{vtvalbe},~\eqref{wvalbe}
can be found in~\cite[Section 17]{suris03}.

\begin{remark}
Svinin~\cite{svinin11} constructed the equation
\begin{equation}
\lb{eqsv}
v_t=\prod_{j=1}^p\frac{1}{v_{-j}-v_{-j+p+1}}
\end{equation}
and the MT
\begin{equation}
\lb{mttsv}
u=\prod_{j=p}^{2p}\frac{1}{v_{-j}-v_{-j+p+1}}
\end{equation}
connecting~\eqref{eqsv} with~\eqref{bog}.
It would be interesting to find out whether the MT~\eqref{mttsv} 
can be constructed by the method of the present paper.
\end{remark}

\section{MTs for the Toda lattice in the Flaschka-Manakov coordinates}
\label{sectoda}

The Toda lattice equation for a scalar function $q=q(n,t)$ 
\begin{equation}
\lb{toda}
q_{tt}=\exp(q_{1}-q)-\exp(q-q_{-1}),\qquad\quad
q_1=q(n+1,t),\qquad q_{-1}=q(n-1,t),
\end{equation}
can be rewritten in the evolution form~\eqref{sdde} 
for a two-component vector-function $u=(u^1,u^2)$ as follows~\cite{flaschka74,manakov75}.
Set $u^1=\exp(q_1-q)$ and $u^2=q_{t}$. Then~\eqref{toda} implies  
\begin{equation}
\label{todamf}
u^1_t=u^1(u^2_{1}-u^2),\qquad\quad
u^2_t=u^1-u^1_{-1}.
\end{equation}
System~\eqref{todamf} is called the Toda lattice in the Flaschka-Manakov coordinates.

It is known (see, e.g.,~\cite{kmw}) that the following matrices form a 
DLR for~\eqref{todamf}
\begin{equation}
\lb{todalr}
M=  \begin{pmatrix} 
\lambda+u^2_1&  u^1 \\ -1 & 0 
\end{pmatrix},\qquad\quad
\lt=  
\begin{pmatrix}
0 &-u^1\\
1&\lambda+u^2
\end{pmatrix}.
\end{equation}
The equation $\pd_t(M)=\cS(\lt)M-M\lt$ is equivalent to~\eqref{todamf}.

To construct MTs by means of the method described in Section~\ref{secgt}, 
we need a DLR $M$, $\lt$ such that $M$ depends only on $u^i$ and $\lambda$.
The matrix $M$ in~\eqref{todalr} is not of this type, because it depends on $u^2_1$.
But this can be easily overcome as follows.

We relabel $u^2:=u^2_1$ and $u^1:=u^1$. 
After this change of variables, 
system~\eqref{todamf} and the DLR~\eqref{todalr} become
\begin{gather}
\label{todanew}
u^1_t=u^1(u^2-u^2_{-1}),\qquad\quad
u^2_t=u^1_1-u^1,\\
\lb{lrnew}
M(u^1,u^2,\lambda)= \begin{pmatrix} 
\lambda+u^2 & u^1 \\ -1 & 0 
\end{pmatrix},\qquad\quad
\lt=  
\begin{pmatrix}
0 &-u^1\\
1&\lambda+u^2_{-1}
\end{pmatrix}.
\end{gather}

According to~\eqref{mumu}, for the DLR~\eqref{lrnew}, 
the group $\hg{1}$ is generated by the matrix-functions 
\begin{equation}
\lb{todah1}
M({\tilde u}^1,{\tilde u}^2,\lambda)\cdot
M(u^1,u^2,\lambda)^{-1}
=\begin{pmatrix} 
\frac{{\tilde u}^1}{u^1} & 
\frac{{\tilde u}^1u^2}{u^1}-{\tilde u}^2+\lambda\big(\frac{{\tilde u}^1}{u^1}-1\big)\\ 
0 & 1 
\end{pmatrix}.
\end{equation}
Note that the right-hand side of~\eqref{todah1} can be simplified by the following conjugation
$$
\begin{pmatrix}
 1 & \lambda \\
 0 & 1
\end{pmatrix}
\begin{pmatrix} 
\frac{{\tilde u}^1}{u^1} & 
\frac{{\tilde u}^1u^2}{u^1}-{\tilde u}^2+\lambda\big(\frac{{\tilde u}^1}{u^1}-1\big)\\ 
0 & 1 
\end{pmatrix}
\begin{pmatrix}
 1 & \lambda \\
 0 & 1
\end{pmatrix}^{-1}=
\begin{pmatrix} 
\frac{{\tilde u}^1}{u^1} & 
\frac{{\tilde u}^1u^2}{u^1}-{\tilde u}^2\\ 
0 & 1 
\end{pmatrix}.
$$
Therefore, it makes sense to consider the following gauge-equivalent DLR
\begin{gather}
\lb{hmtoda}
\hat M=
\begin{pmatrix}
 1 & \lambda \\
 0 & 1
\end{pmatrix}
\cdot M(u^1,u^2,\lambda)\cdot
\begin{pmatrix}
 1 & \lambda \\
 0 & 1
\end{pmatrix}^{-1}=
\begin{pmatrix}
 u^2 & u^1-u^2 \lambda  \\
 -1 & \lambda  \\
\end{pmatrix},\\
\lb{hutoda}
\hat\lt=
\begin{pmatrix}
 1 & \lambda \\
 0 & 1
\end{pmatrix}
\cdot\lt\cdot
\begin{pmatrix}
 1 & \lambda \\
 0 & 1
\end{pmatrix}^{-1}=
\begin{pmatrix}
 \lambda  & \lambda u^2_{-1}-u^1 \\
 1 & u^2_{-1} \\
\end{pmatrix}.
\end{gather}
For this DLR, the group $\hg{1}$ consists of the constant matrix-functions 
$\begin{pmatrix}
 a_1 & a_2 \\
 0 & 1 
\end{pmatrix}$
for all $a_1,a_2\in\fik$, $a_1\neq 0$.

Let us construct some MTs for~\eqref{todanew}, using Theorem~\ref{tharb}.
Recall that $G$ is given by~\eqref{ggr}.
Using the notation from Section~\ref{sgalr}, 
consider the case $\sm=2$, $\og=1$, $k_1=1$, $\cp_1\in\fik$.
Then 
\begin{equation}
\lb{gtoda}
G=\gq{2}{\cp_1}{1}
=\left.\left\{
\begin{pmatrix}
 w^{10}_{11} & w^{10}_{12} \\
 w^{10}_{21} & w^{10}_{22} 
\end{pmatrix}
\ \right|\ 
w^{10}_{11}w^{10}_{22}-w^{10}_{21}w^{10}_{12}\neq 0\right\}
\cong\GL_2(\fik).
\end{equation}
We set also $H=\iel$, where $\iel\in G$ is the identity element.
Since in the product
$$
\begin{pmatrix}
 a_1 & a_2 \\
 0 & 1 
\end{pmatrix}\cdot
\begin{pmatrix}
 w^{10}_{11} & w^{10}_{12} \\
 w^{10}_{21} & w^{10}_{22} 
\end{pmatrix}
$$
the $w^{10}_{21}$ and $w^{10}_{22}$ components do not change, 
the functions $z^1=w^{10}_{21}$ and $z^2=w^{10}_{22}$ are $\hg{1}$-left-invariant.
As the subgroup $H\subset G$ is trivial, these functions are also $H$-right-invariant.

For the DLR~\eqref{hmtoda},~\eqref{hutoda} and the group~\eqref{gtoda}, 
formulas~\eqref{csw},~\eqref{dtw} become
\begin{gather}
\lb{swtoda}
\cS\begin{pmatrix}
 w^{10}_{11} & w^{10}_{12} \\
 w^{10}_{21} & w^{10}_{22} 
\end{pmatrix}
=\begin{pmatrix}
 u^2 & u^1-u^2 \lambda  \\
 -1 & \lambda  
\end{pmatrix}
\begin{pmatrix}
 w^{10}_{11} & w^{10}_{12} \\
 w^{10}_{21} & w^{10}_{22} 
\end{pmatrix},
\qquad \la-\cp_1=0,\\
\lb{dtwtoda}
D_t\begin{pmatrix}
 w^{10}_{11} & w^{10}_{12} \\
 w^{10}_{21} & w^{10}_{22} 
\end{pmatrix}
=
\begin{pmatrix}
 \lambda  & \lambda u^2_{-1}-u^1 \\
 1 & u^2_{-1} 
\end{pmatrix}
\begin{pmatrix}
 w^{10}_{11} & w^{10}_{12} \\
 w^{10}_{21} & w^{10}_{22} 
\end{pmatrix},
\qquad \la-\cp_1=0.
\end{gather}

Using formula~\eqref{swtoda} and the notation $z^1_k=\cS^k(z^1)$, $z^2_k=\cS^k(z^2)$ for $k\in\zz$, 
we obtain
\begin{gather}
\lb{z101}
z^1_0=z^1=w^{10}_{21},\quad\qquad
z^1_1=\cS(z^1)=\cS(w^{10}_{21})=\cp_1 w^{10}_{21}-w^{10}_{11},\\
\lb{z12}
z^1_2=\cS^2(z^1)=\cS(z^1_1)=
\cp_1^2w^{10}_{21}-\cp_1w^{10}_{11}-u^2w^{10}_{11}-u^1w^{10}_{21}+\cp_1u^2w^{10}_{21},\\
\lb{z201}
z^2_0=z^2=w^{10}_{22},\qquad\quad
z^2_1=\cS(z^2)=\cS(w^{10}_{22})=\cp_1 w^{10}_{22}-w^{10}_{12},\\
\lb{z22}
z^2_2=\cS^2(z^2)=\cS(z^2_1)=
\cp_1^2w^{10}_{22}-\cp_1w^{10}_{12}-u^2w^{10}_{12}-u^1w^{10}_{22}+\cp_1u^2w^{10}_{22}.
\end{gather}
Equations~\eqref{z101},~\eqref{z201} imply that $z^1_0$, $z^1_1$, $z^2_0$, $z^2_1$ 
form a system of coordinates for~$G$, so condition~\eqref{msclc} is valid for $\da_1=\da_2=1$.
Formulas~\eqref{z12},~\eqref{z22} show that condition~\eqref{mpdu} is valid as well.

From~\eqref{z101}, \eqref{z12}, \eqref{z201}, \eqref{z22} one can express $u^1$, $u^2$ in terms of 
$z^1_0$, $z^1_1$, $z^2_0$, $z^2_1$, $z^1_{2}$, $z^2_{2}$ as follows
\begin{equation} 
\lb{u12toda}
u^1=\frac{z^1_{2}z^2_{1}-z^1_{1}z^2_{2}}{z^1_{1} z^2_0-z^1_0 z^2_{1}},\qquad\quad
u^2=-\cp_1+
\frac{z^1_0z^2_{2}-z^1_{2}z^2_0}{z^1_0z^2_{1}-z^1_{1}z^2_0}.
\end{equation}
Using formulas~\eqref{dtwtoda},~\eqref{u12toda}, we can compute $D_t(z^1)$, $D_t(z^2)$ and get 
\begin{equation}
\lb{dtz12}
D_t(z^1)=\frac{z^1_{-1}(z^1_0 z^2_{1}-z^1_{1} z^2_0)}{z^1_{-1} z^2_0-z^1_0 z^2_{-1}},\qquad\quad
D_t(z^2)=\frac{z^2_{-1}(z^1_0 z^2_{1}-z^1_{1} z^2_0)}{z^1_{-1} z^2_0-z^1_0 z^2_{-1}}.
\end{equation}

According to Theorem~\ref{tharb}, to obtain an MT, 
we need to replace $z^i_k$ by $v^i_k$ for all $i=1,2$ and $k\in\zz$ in~\eqref{u12toda},~\eqref{dtz12}.
(And we can use the identification $v^1_0=v^1$, $v^2_0=v^2$.) 

This gives the following result. For any $\cp_1\in\fik$, the formulas
\begin{equation}
\lb{fmtu12}
u^1=\frac{v^1_{2}v^2_{1}-v^1_{1} v^2_{2}}{v^1_{1} v^2-v^1 v^2_{1}},\qquad\quad
u^2=-\cp_1+ 
\frac{v^1v^2_{2}-v^1_{2} v^2}{v^1 v^2_{1}-v^1_{1} v^2}
\end{equation}
determine an MT from 
\begin{equation}
\lb{fmev12}
v^1_t=\frac{v^1_{-1}(v^1 v^2_{1}-v^1_{1} v^2)}{v^1_{-1} v^2-v^1 v^2_{-1}},\qquad\quad
v^2_t=\frac{v^2_{-1}(v^1 v^2_{1}-v^1_{1} v^2)}{v^1_{-1} v^2-v^1 v^2_{-1}}
\end{equation}
to~\eqref{todanew}.
We cannot find~\eqref{fmtu12},~\eqref{fmev12} in the existing literature, 
so this MT may be new.

According to the above notation, $v^1=z^1=w^{10}_{21}$ and 
$v^2=z^2=w^{10}_{22}$, where $w^{10}_{21}$, $w^{10}_{22}$ satisfy \eqref{swtoda}, 
\eqref{dtwtoda}. Therefore, solutions of system~\eqref{fmev12} can be obtained 
from solutions of system~\eqref{todanew} as follows.
For a given solution $u^1(n,t)$, $u^2(n,t)$ of system~\eqref{todanew}, one needs 
to solve the auxiliary linear system \eqref{swtoda}, \eqref{dtwtoda}
for the functions $w^{10}_{ij}(n,t)$, $i,j=1,2$, and then 
one can take $v^1(n,t)=w^{10}_{21}(n,t)$, $v^2(n,t)=w^{10}_{22}(n,t)$.

For example, consider the constant solution 
$u^1(n,t)=u^2(n,t)=1$ of~\eqref{todanew}, and set $\la=\cp_1=1$.
Then, solving \eqref{swtoda}, \eqref{dtwtoda} in the case 
$u^1(n,t)=u^2(n,t)=\la=1$, we obtain 
$$
w^{10}_{11}=b_1\mathrm{e}^{t},\quad
w^{10}_{12}=b_2\mathrm{e}^{t},\quad
w^{10}_{21}=(b_1t+b_3-nb_1)\mathrm{e}^{t},\quad
w^{10}_{22}=(b_2t+b_4-nb_2)\mathrm{e}^{t}, 
$$
where $b_1$, $b_2$, $b_3$, $b_4$ are arbitrary constants.
Then 
$$
v^1(n,t)=w^{10}_{21}=(b_1t+b_3-nb_1)\mathrm{e}^{t},\qquad 
v^2(n,t)=w^{10}_{22}=(b_2t+b_4-nb_2)\mathrm{e}^{t}
$$
is a solution of system~\eqref{fmev12}.

\section{MTs for Adler-Postnikov lattices}
\label{secap}

In what follows, \emph{difference operators} are polynomials in $\cS$, $\cS^{-1}$.

Adler and Postnikov~\cite{adler-pos11} studied 
integrable hierarchies of differential-difference equations
associated with spectral problems of the form
\begin{equation}
\lb{ppq}
P\psi=\lambda Q\psi,
\end{equation}
where $P$, $Q$ are difference operators and $\psi=\psi(n,t)$ is a scalar function.

Such hierarchies are constructed in~\cite{adler-pos11} as follows.
One considers the Lax type equations
\begin{equation}
\lb{pabq}
P_{t}=BP-PA,\qquad\quad Q_{t}=BQ-QA
\end{equation}
for some difference operators $P$, $Q$, $A$, $B$.
Then the equation 
\begin{equation}
\lb{psita}
\psi_t=A\psi
\end{equation}
is compatible with~\eqref{ppq} modulo~\eqref{pabq}.
For certain fixed operators $P$ and $Q$, Adler and Postnikov~\cite{adler-pos11}
find an infinite collection of operators $A$, $B$ so that~\eqref{pabq} 
becomes a commutative hierarchy of differential-difference equations.

Recall that in Section~\ref{smcmt} we have described a method to construct MTs 
for differential-difference equations possessing DLRs of the form~\eqref{mu}.
This method is applicable to equations presented in~\cite{adler-pos11}. 
To illustrate this, let us consider the following operators from~\cite[Section 5]{adler-pos11}
\begin{gather}
\lb{pqex}
P=u^1\cS^3+\cS,\qquad\quad Q=\cS^2+u^2,\\
\lb{abmin}
A^{-}=u^2_{-2}u^2_{-1}S^{-2}+\fap_{-3}+\fap_{-2},\qquad\quad B^{-}=u^2_{-1}u^2S^{-2}+\fap_{-1}+\fap,\\
\lb{abplus}
A^{+}=u^1_{-2}u^1_{-1}S^{2}+\gap_{-1}+\gap,\qquad\quad 
B^{+}=u^1u^1_1S^{2}+\gap+\gap_1,
\end{gather}
where 
$\fap=u^1u^2_1u^2_2,\quad \gap=u^1_{-2}u^1_{-1}u^2,\quad \fap_k=\cS^k(\fap),\quad 
\gap_k=\cS^k(\gap)$ for $k\in\zz$.

Following~\cite{adler-pos11}, we consider equations~\eqref{pabq},~\eqref{psita} 
in the case $A=A^{-}$, $B=B^{-}$ and in the case $A=A^{+}$, $B=B^{+}$.
In these cases, the variable $t$ is denoted by $t^{-}$ and $t^{+}$ respectively.

The system $P_{t^{-}}=B^{-}P-PA^{-},\,\ Q_{t^{-}}=B^{-}Q-QA^{-}$ is equivalent to 
\begin{equation}
\lb{utmin}
\begin{aligned}
u^1_{t^{-}} & = u^1(\fap_{-1}-\fap_{1}),\\
u^2_{t^{-}} & = u^2(\fap+\fap_{-1}-\fap_{-2}-\fap_{-3}-u^2_1+u^2_{-1}).
\end{aligned}
\end{equation}
The system $P_{t^{+}}=B^{+}P-PA^{+},\,\ Q_{t^{+}}=B^{+}Q-QA^{+}$ is equivalent to 
\begin{equation}
\lb{utplus}
\begin{aligned}
u^1_{t^{+}} & = u^1(\gap+\gap_1-\gap_2-\gap_3-u^1_{-1}+u^1_1),\\
u^2_{t^{+}} & = u^2(\gap_1-\gap_{-1}).
\end{aligned}
\end{equation}

Let us construct matrix DLRs for 
the differential-difference equations~\eqref{utmin} and~\eqref{utplus}.
Set 
\begin{equation}
\lb{vfps}
\vf^1=\psi,\qquad\quad
\vf^2=\cS(\psi),\qquad\quad
\vf^3=\cS^2(\psi).
\end{equation}
Using \eqref{ppq}, \eqref{pqex}, and \eqref{vfps}, one gets
\begin{equation}
\lb{svf}
\cS
\begin{pmatrix}
 \vf^1 \\
 \vf^2 \\
 \vf^3 \\
\end{pmatrix}
=
\begin{pmatrix}
 0 & 1 & 0 \\
 0 & 0 & 1 \\
 \lambda \frac{u^2}{u^1} & -\frac{1}{u^1} & \frac{\lambda}{u^1}  \\
\end{pmatrix}
\begin{pmatrix}
\vf^1 \\
 \vf^2 \\
 \vf^3 \\
\end{pmatrix}.
\end{equation}
Using formulas~\eqref{abmin},~\eqref{vfps} and equation~\eqref{psita} for $A=A^{-}$, $t=t^{-}$, 
we obtain
\begin{gather}
\lb{dtmvf1}
\pd_{t^{-}}(\vf^1)=\psi_{t^{-}}=A^{-}\psi=u^2_{-2}u^2_{-1}\cS^{-2}(\psi)+(\fap_{-3}+\fap_{-2})\psi,\\
\lb{dtmvf2}
\pd_{t^{-}}(\vf^2)=\cS(\psi_{t^{-}})=u^2_{-1}u^2\cS^{-1}(\psi)+(\fap_{-2}+\fap_{-1})\cS(\psi),\\
\lb{dtmvf3}
\pd_{t^{-}}(\vf^3)=\cS^2(\psi_{t^{-}})=u^2u^2_{1}\psi+(\fap_{-1}+\fap)\cS^2(\psi).
\end{gather}
Applying the operators $\cS^{-1}$ and $\cS^{-2}$ to~\eqref{svf}, one can express 
the functions $\cS^{-1}(\psi)=\cS^{-1}(\vf^1)$ and $\cS^{-2}(\psi)=\cS^{-2}(\vf^1)$ 
in terms of $\vf^1$, $\vf^2$, $\vf^3$. 
Substituting these expressions and $\psi=\vf^1$ in the right-hand sides 
of~\eqref{dtmvf1}, \eqref{dtmvf2}, \eqref{dtmvf3}, we get 
\begin{equation}
\lb{dtminvf}
\pd_{t^{-}}
\begin{pmatrix}
 \vf^1 \\
 \vf^2 \\
 \vf^3 \\
\end{pmatrix}
=
\begin{pmatrix}
 \fap_{-3}+\fap_{-2}-u^2_{-1}+\frac{1}{\lambda^2} & \frac{u^1_{-2} u^2_{-1}-1}{\lambda } & \frac{u^1_{-1}}{\lambda^2} \\
 \frac{u^2}{\lambda } & -u^2+\fap_{-2}+\fap_{-1} & \frac{u^2 u^1_{-1}}{\lambda } \\
 u^2 u^2_1 & 0 & \fap+\fap_{-1} \\
\end{pmatrix}
\begin{pmatrix}
 \vf^1 \\
 \vf^2 \\
 \vf^3 \\
\end{pmatrix}.
\end{equation}

The equations $P\psi=\lambda Q\psi$ and $\psi_{t^{-}}=A^{-}\psi$
are compatible modulo the system
$$
P_{t^{-}}=B^{-}P-PA^{-},\qquad\quad Q_{t^{-}}=B^{-}Q-QA^{-},
$$
which is equivalent to~\eqref{utmin}.
Since~\eqref{svf} and \eqref{dtminvf} have been obtained from the equations
$P\psi=\lambda Q\psi$, $\psi_{t^{-}}=A^{-}\psi$, 
we see that system~\eqref{svf}, \eqref{dtminvf} is compatible modulo~\eqref{utmin}.

Therefore, the matrices 
\begin{equation}
\lb{lrtmin}
\begin{aligned}
M&=
\begin{pmatrix}
 0 & 1 & 0 \\
 0 & 0 & 1 \\
 \lambda \frac{u^2}{u^1} & -\frac{1}{u^1} & \frac{\lambda}{u^1}  \\
\end{pmatrix},\\
\lt^{-}&=
\begin{pmatrix}
 \fap_{-3}+\fap_{-2}-u^2_{-1}+\frac{1}{\lambda^2} & \frac{u^1_{-2} u^2_{-1}-1}{\lambda } & \frac{u^1_{-1}}{\lambda^2} \\
 \frac{u^2}{\lambda } & -u^2+\fap_{-2}+\fap_{-1} & \frac{u^2 u^1_{-1}}{\lambda } \\
 u^2 u^2_1 & 0 & \fap+\fap_{-1} \\
\end{pmatrix}
\end{aligned}
\end{equation}
form a DLR for~\eqref{utmin}. The equation
$\pd_{t^{-}}(M)=\cS(\lt^{-})M-M\lt^{-}$ is equivalent to~\eqref{utmin}.

Similarly, using the equations $P\psi=\lambda Q\psi$ and $\psi_{t^{+}}=A^{+}\psi$, 
one gets the following DLR for~\eqref{utplus}
\begin{equation}
\lb{lrtplus}
M=
\begin{pmatrix}
 0 & 1 & 0 \\
 0 & 0 & 1 \\
 \lambda \frac{u^2}{u^1} & -\frac{1}{u^1} & \frac{\lambda}{u^1}  \\
\end{pmatrix},\qquad
\lt^{+}=
\begin{pmatrix}
 \gap+\gap_{-1} & 0 & u^1_{-2} u^1_{-1} \\
 \lambda u^2 u^1_{-1} & \gap+\gap_1-u^1_{-1} & \lambda  u^1_{-1} \\
 \lambda^2 u^2 & \lambda\left(u^1 u^2_1-1\right) & \lambda^2-u^1+\gap_1+\gap_2 \\
\end{pmatrix}.
\end{equation}

Recall that the groups $\hg{k}$, $k\in\zp$, 
corresponding to a DLR~\eqref{mu} are determined by the matrix $M$.
Using $M$ from~\eqref{lrtmin},~\eqref{lrtplus}, 
one obtains the following results on the corresponding groups $\hg{1}$, $\hg{2}$.
\begin{itemize}
\item The group $\hg{1}$ consists of the matrix-functions 
$\begin{pmatrix}
 1 & 0 & 0 \\
 0 & 1 & 0 \\
 a_1 & -a_1 \lambda  & a_2 \\
\end{pmatrix}$
for all $a_1,a_2\in\fik$, $a_2\neq 0$.
\item The group $\hg{2}$ consists of the matrix-functions 
$\begin{pmatrix}
 1 & 0 & 0 \\
\al_{21}(\la) & \al_{22}(\la) & \al_{23}(\la) \\
\al_{31}(\la) & \al_{32}(\la) & \al_{33}(\la) \\ 
\end{pmatrix}$
for some scalar functions $\al_{ij}(\la)$, $i=2,3$, $j=1,2,3$. 
(We do not need the explicit form of these functions.)
\end{itemize}

Let us construct some MTs for systems~\eqref{utmin} and~\eqref{utplus},
using Theorem~\ref{tharb} and Remark~\ref{remact}.

Since \eqref{utmin} and~\eqref{utplus} are two-component systems, we have $\diu=2$.
Using the notation from Section~\ref{sgalr}, consider the case 
$$
\sm=3,\ \og=2,\ k_1=k_2=1,\,\ 
G=\gq{3}{\cp_1}{1}\times\gq{3}{\cp_2}{1}\cong\GL_3(\fik)\times\GL_3(\fik),
\,\ 
\cp_1,\cp_2\in\fik.
$$
The standard transitive actions of $\GL_3(\fik)$ on the manifolds $\wor{3}$
and $\prs^{2}$ give a transitive action of ${G\cong\GL_3(\fik)\times\GL_3(\fik)}$ 
on $\mnf=(\wor{3})\times\prs^{2}$.
We take $H\subset G$ to be the stabilizer subgroup of a point~$\pmm\in\mnf$, 
so $H$ is given by~\eqref{gxx}. Then $G/H\cong\mnf$.

We are going to apply Theorem~\ref{tharb} in the case $\da_1=2$, $\da_2=1$.
According to Remark~\ref{remact} and Theorem~\ref{tharb}, 
we need to find an $\hg{2}$-left-invariant function~$z^1$ 
and an $\hg{1}$-left-invariant function~$z^2$
on an open dense subset of the manifold $(\wor{3})\times\prs^{2}$
such that conditions~\eqref{msclc},~\eqref{mpdu} hold for 
$G/H\cong(\wor{3})\times\prs^{2}$ and $r=\dim G/H=5$.

Let $\al^1$, $\al^2$, $\al^3$ be coordinates on $\wor{3}$ 
and $\beta^1$, $\beta^2$, $\beta^3$ be 
homogeneous coordinates on the projective space $\prs^{2}$.
The above results on the structure of $\hg{1}$ and $\hg{2}$ imply that 
the function $z^1=\al^1$ is $\hg{2}$-left-invariant 
and the function $z^2=\beta^1/\beta^2$ is $\hg{1}$-left-invariant.

A straightforward computation shows that conditions~\eqref{msclc},~\eqref{mpdu} are satisfied, 
so we can apply Theorem~\ref{tharb}. 
This gives the following result. 
For any $\cp_1,\cp_2\in\fik$, the formulas
\begin{equation}
\lb{u12pq}
u^1= \frac{v^2_2 \left(\cp_2 v^2 v^2_1 \left(v^1_1-\cp_1 v^1_2\right)+\cp_1 v^1 \left(\cp_2-v^2_1\right)\right)}{\cp_1 v^1-\cp_2 v^1_3 v^2 v^2_1 v^2_2},\qquad\quad
u^2= \frac{-\cp_1 v^1_2+v^1_3 v^2_2 \left(\cp_2-v^2_1\right)+v^1_1}{\cp_1 v^1-\cp_2 v^1_3 v^2 v^2_1 v^2_2}
\end{equation}
determine an MT from the equation 
\begin{equation}
\lb{ptm}
\begin{aligned}
v^1_{t^{-}} & = v^1 \left(u^1_{-2} u^1_{-1} u^2+u^1_{-3} u^1_{-2} u^2_{-1}\right)+v^1_2 u^1_{-2} u^1_{-1},\\ 
v^2_{t^{-}} & = \frac{-\cp_2 v^2_1 (v^2)^2 u^1_{-1} u^2+v^2 \left(v^2_1 u^1_{-3} u^1_{-2} u^2_{-1}-u^1_{-1} \left(\cp_2+v^2_1 \left(u^1 u^2_1-1\right)\right)\right)+u^1_{-2} u^1_{-1}}{v^2_1} 
\end{aligned}
\end{equation}
to~\eqref{utmin}, and from the equation
\begin{equation}
\lb{ptpl}
\begin{aligned}
v^1_{t^{+}} & = \frac{\cp_1 v^1_1 \left(u^1_{-2} u^2_{-1}-1\right)+v^1 \left(\cp_1^2 u^2_{-1} \left(u^1_{-2} u^2+u^1_{-3} u^2_{-2}-1\right)+1\right)+v^1_2 u^1_{-1}}{\cp_1^2},\\
v^2_{t^{+}} & =  \frac{ -\cp_2 (v^2)^2 u^2+v^2 \left(\cp_2^2 u^2_{-1} \left(u^1_{-3} u^2_{-2}-1\right)+\cp_2^2 u^2 \left(1-u^1_{-1} u^2_1\right)+1\right)}{\cp_2^2 }+\\ 
&\qquad +\frac{v^2_1 \cp_2 \left(u^1_{-2} u^2_{-1}-1\right)+u^1_{-1} \left(1-\cp_2 v^2 u^2\right)}{\cp_2^2 v^2_1} 
\end{aligned}
\end{equation}
to~\eqref{utplus}.
In formulas~\eqref{ptm},~\eqref{ptpl}, 
one has $u^1_k=\cS^k(u^1)$ and $u^2_k=\cS^k(u^2)$ for $k\in\zz$, 
where $u^1$, $u^2$ are given by~\eqref{u12pq}. For example, according to~\eqref{u12pq},
$$
u^2_{-1}=\cS^{-1}(u^2)=\cS^{-1}\Big(
\frac{-\cp_1 v^1_2+v^1_3 v^2_2 \left(\cp_2-v^2_1\right)+v^1_1}{\cp_1 v^1-\cp_2 v^1_3 v^2 v^2_1 v^2_2}
\Big)=
\frac{-\cp_1 v^1_1+v^1_2 v^2_1 \left(\cp_2-v^2\right)+v^1}{\cp_1 v^1_{-1}-\cp_2 v^1_2 v^2_{-1} v^2 v^2_1}.
$$

\section{Conclusion}

In this paper, we have presented a geometric method to construct 
Miura-type transformations (MTs) 
and modified equations for differential-difference (lattice) equations, 
including multicomponent ones.
We construct MTs and modified equations 
from invariants of certain Lie group actions on manifolds associated with 
matrix Darboux-Lax representations (DLRs) of differential-difference equations.

Using this construction, from a given suitable DLR one can obtain many MTs of different orders.
As has been shown in Example~\ref{bcpcp},
the order of the obtained MTs may be higher than the size of the matrices in the DLR.

Applying this method to DLRs of the Volterra, 
Narita-Itoh-Bogoyavlensky, Toda lattices, 
and Adler-Postnikov lattices from~\cite{adler-pos11},
we have constructed a number of MTs and modified lattice equations.
The MTs \eqref{uv2v3}, \eqref{ucpvcp}, \eqref{fmtu12}, \eqref{u12pq}
and modified equations 
\eqref{vtvv2}, \eqref{vtvcpcp}, \eqref{fmev12}, \eqref{ptm}, \eqref{ptpl} seem to be new.

The described method to construct MTs can help in solving classification problems 
for integrable differential-difference equations.
Indeed, when one tries to classify a certain class of such equations, 
one often finds a few basic equations such that all other equations 
from the considered class can be obtained from 
the basic ones by means of MTs~\cite{yam2006}.
Therefore, systematic methods to construct MTs can help 
in obtaining such classification results.

In Section~\ref{secgt}, for any given DLR of the form~\eqref{mu} 
we have defined an increasing sequence of groups $\hg{k}$, $k\in\zp$.
As has been shown in Section~\ref{secgt}, 
these groups play a crucial role in the construction of MTs.
In Sections~\ref{secgt},~\ref{secbog},~\ref{sectoda},~\ref{secap}, 
we have computed $\hg{k}$ for some small values of~$k$ 
for a number of DLRs, and this has allowed us to construct the above-mentioned MTs.

It would be interesting to compute the groups $\hg{k}$ for higher values of $k$, 
because this may give many new kinds of MTs.
Also, it would be interesting to try to describe all MTs 
that can be constructed by the presented method for some well-known DLRs 
(such as the DLRs of the Narita-Itoh-Bogoyavlensky and Toda lattices).

It is well known that auto-B\"acklund transformations can often 
be obtained as compositions of MTs. Therefore, it is natural 
to ask the following question: 
when can one obtain auto-B\"acklund transformations, using 
MTs constructed by the presented method?

We leave these problems for future work.

\section*{Acknowledgements}
The authors would like to thank V.~E.~Adler, A.~V.~Mikhailov, Yu.~B.~Suris, P.~Xenitidis, 
and R.~I.~Yamilov for useful discussions.
The authors gratefully acknowledge support from a Leverhulme Trust grant.


\begin{thebibliography}{99}

\bibitem{adler-pos11}
V.~E.~Adler and V.~V.~Postnikov.
Differential-difference equations associated with the fractional Lax operators.
\emph{J. Phys. A} \textbf{44} (2011), 415203, 17 pp. 

\bibitem{bogoy88}
O.~I.~Bogoyavlensky.
Integrable discretizations of the KdV equation.
\emph{Physics Letters A} \textbf{134} (1988), 34--38.

\bibitem{drin-sok85}
V.~G.~Drinfeld and V.~V.~Sokolov.
Equations that are related to the Korteweg-de Vries equation.
\emph{Dokl. Akad. Nauk SSSR} \textbf{284} (1985), 29--33 (in Russian). 

V.~G.~Drinfeld and V.~V.~Sokolov. 
\emph{Soviet Math. Dokl.} \textbf{32} (1985), 361--365 (Engl. transl.). 

\bibitem{flaschka74}
H.~Flaschka. The Toda lattice. I. Existence of integrals. 
\emph{Phys. Rev. B (3)} \textbf{9} (1974), 1924--1925.

\bibitem{gram11}
B.~Grammaticos, A.~Ramani, C.~Scimiterna, and R.~Willox.
Miura transformations and the various guises of integrable lattice equations. 
\emph{J. Phys. A} \textbf{44} (2011), 152004, 9 pp. 

\bibitem{igon05}
S.~A.~Igonin. Miura type transformations and homogeneous spaces. 
\emph{J. Phys. A} \textbf{38} (2005), 4433--4446.

\bibitem{itoh87}
Y.~Itoh.
Integrals of a Lotka-Volterra system of odd number of variables.
\emph{Progr. Theoret. Phys.} \textbf{78} (1987), 507--510. 

\bibitem{kmw}
F.~Khanizadeh, A.~V.~Mikhailov, and Jing Ping Wang.
Darboux transformations and recursion operators for differential-difference equations.
\emph{Theoret. and Math. Phys.} \textbf{177} (2013), 1606--1654.

\bibitem{manakov75}
S.~V.~Manakov. 
Complete integrability and stochastization of discrete dynamical systems. 
\emph{Soviet Physics JETP} \textbf{40} (1975), 269--274.

\bibitem{meshk2008}
A.~G.~Meshkov and M.~Ju.~Balakhnev. 
Two-field integrable evolutionary systems of the third order and their differential substitutions. \emph{SIGMA} \textbf{4} (2008), Paper 018, 29 pp.

\bibitem{mss91} A.~V.~Mikhailov, A.~B.~Shabat, and V.~V.~Sokolov. 
The symmetry approach to classification of integrable equations. 
In: \emph{What is integrability?}, edited by V.~E.~Zakharov, Springer--Verlag, 1991.

\bibitem{mx14}
A.~V.~Mikhailov and P.~Xenitidis. 
Second order integrability conditions for difference equations: an integrable equation. 
\emph{Lett. Math. Phys.} \textbf{104} (2014), 431--450.

\bibitem{miura-cl}
R.~M.~Miura, C.~S.~Gardner, and M.~D.~Kruskal. 
Korteweg-de Vries equation and generalizations. II. Existence of conservation laws and constants of motion. \emph{J. Math. Phys.} \textbf{9} (1968), 1204--1209. 

\bibitem{narita82}
K.~Narita. 
Soliton solution to extended Volterra equation. 
\emph{J. Phys. Soc. Japan} \textbf{51} (1982), 1682--1685.

\bibitem{levi14}
C.~Scimiterna, M.~Hay, and D.~Levi.
On the integrability of a new lattice equation found by multiple scale analysis.
\emph{J. Phys. A} \textbf{47} (2014), 265204, 16 pp. 
arXiv:1401.5691v1

\bibitem{sokolov88}
V.~V.~Sokolov.
Pseudosymmetries and differential substitutions.
\emph{Funkts. Anal. Prilozh.} \textbf{22} (1988), 47--56 (in Russian).

V.~V.~Sokolov. 
\emph{Funct. Anal. Appl.} \textbf{22} (1988), 121--129 (Engl. transl.).

\bibitem{svinin11}
A.~K.~Svinin.
On some integrable lattice related by the Miura-type transformation 
to the Itoh-Narita-Bogoyavlenskii lattice. 
\emph{J. Phys. A} \textbf{44} (2011), 465210, 8 pp.

\bibitem{suris03}
Yu.~B.~Suris. 
\emph{The Problem of Integrable Discretization: Hamiltonian Approach.} 
Progress in Mathematics, 219. Birkh\"auser Verlag, Basel, 2003.

\bibitem{yam94}
R.~I.~Yamilov.
Construction scheme for discrete Miura transformations.
\emph{J. Phys. A} \textbf{27} (1994), 6839--6851. 

\bibitem{yam2006}
R.~Yamilov.
Symmetries as integrability criteria for differential difference equations.
\emph{J. Phys. A} \textbf{39} (2006), R541--R623. 

\end{thebibliography}
\end{document}